\documentclass[11pt]{article}
\usepackage{amsthm}
\usepackage{amsfonts}
\usepackage{amsmath}
\usepackage{amssymb}
\usepackage{colonequals}
\usepackage{epsfig}
\usepackage[utf8]{inputenc}

\newtheorem{theorem}{Theorem}

\newtheorem{lemma}[theorem]{Lemma}

\newcommand{\tw}{\text{tw}}
\newcommand{\ins}{\text{ins}}
\newcommand{\dom}{\text{dom}}
\newcommand{\TT}{{\cal T}}
\newcommand{\OO}{{\cal O}}

\newcommand{\R}{\text{Rad}}
\newcommand{\rad}{\text{rad}}
\newcommand{\atom}{\text{Atom}}
\newcommand{\id}{\text{id}}
\newcommand{\repr}{\text{repr}}

\title{List-coloring apex-minor-free graphs}
\author{Zdeněk Dvořák\thanks{Computer Science Institute, Charles University, Prague, Czech Republic. E-mail: {\tt rakdver@iuuk.mff.cuni.cz}.
Supported by the Center of Excellence -- Inst. for Theor. Comp. Sci., Prague (project P202/12/G061 of Czech Science Foundation), and
by project LH12095 (New combinatorial algorithms - decompositions, parameterization, efficient solutions) of Czech Ministry of Education.}\and
Robin Thomas\thanks{Georgia Institute of Technology, Atlanta, GA, USA. E-mail: {\tt thomas@math.gatech.edu}. Partially supported by NSF under Grant No.~DMS-1202640.}}
\date{}

\begin{document}
\maketitle

\begin{abstract}
A graph $H$ is \emph{$t$-apex} if $H-X$ is planar for some set $X\subset V(H)$ of size $t$.
For any integer $t\ge 0$ and a fixed $t$-apex graph $H$, we give a polynomial-time algorithm to decide whether
a $(t+3)$-connected $H$-minor-free graph is colorable from a given assignment of lists
of size $t+4$.  The connectivity requirement is the best possible in the sense that
for every $t\ge 1$, there exists a $t$-apex graph $H$ such that
testing $(t+4)$-colorability of $(t+2)$-connected $H$-minor-free graphs is NP-complete.
Similarly, the size of the lists cannot be decreased (unless $\text{P}=\text{NP}$), since for every $t\ge 1$, testing
$(t+3)$-list-colorability of $(t+3)$-connected $K_{t+4}$-minor-free graphs is NP-complete.
\end{abstract}

All graphs considered in this paper are finite and simple.  Let $G$ be a graph.  A function $L$ which assigns a set of colors to each vertex of $G$ is called a \emph{list assignment}.
An \emph{$L$-coloring} $\phi$ of $G$ is a function such that $\phi(v)\in L(v)$ for each $v\in V(G)$ and such that
$\phi(u)\neq\phi(v)$ for each edge $uv\in E(G)$.  For an integer $k$, we say that $L$ is a \emph{$k$-list assignment}
if $|L(v)|=k$ for every $v\in V(G)$, and $L$ is a \emph{$(\ge\!k)$-list assignment} if $|L(v)|\ge k$ for every $v\in V(G)$.

The concept of list coloring was introduced by Vizing~\cite{vizing1976} and Erd\H{o}s et al.~\cite{erdosrubintaylor1979}.
Clearly, list coloring generalizes ordinary proper coloring; a graph has chromatic number at most $k$ if and only if
it can be $L$-colored for the $k$-list assignment which assigns the same list to each vertex.  Consequently, the computational problem of
deciding whether a graph can be colored from a given $k$-list assignment is NP-complete for every $k\ge 3$~\cite{garey1979computers}
(while for $k\le 2$, it is polynomial-time decidable~\cite{erdosrubintaylor1979}).
Let this problem be denoted by $k$-LC.

The complexity of the $k$-LC problem motivates a study of restrictions which ensure that it becomes polynomial-time decidable.
Thomassen~\cite{thomassen1994} proved that every planar graph can be colored from any $(\ge\!5)$-list assignment,
thus showing that $k$-LC is polynomial-time decidable for planar graphs for any $k\ge 5$.  On the other hand,
$k$-LC is NP-complete even for planar graphs for $k\in\{3,4\}$, see~\cite{chooscomplex}.
More generally, for any fixed surface $\Sigma$, the problem $k$-LC with $k\ge 5$ is polynomial-time decidable for
graphs embedded in $\Sigma$~\cite{dvokawalg,lukethe}.  Let us remark that for ordinary coloring,
it is not known whether there exists a polynomial-time algorithm deciding whether a graph embedded in a fixed
surface other than the sphere is $4$-colorable (while all graphs embedded in the sphere are $4$-colorable~\cite{AppHak1,AppHakKoc,rsst}).

In this paper, we study a further generalization of this problem---deciding $k$-LC for graphs from a fixed proper minor-closed family.
Each such family is determined by a finite list of forbidden minors~\cite{rs20}, and thus we consider the complexity
of $k$-LC for graphs avoiding a fixed graph $H$ as a minor.  
A graph $H$ is \emph{$t$-apex} if $H-X$ is planar for some set $X\subset V(H)$ of size $t$.
Our main result is the following.

\begin{theorem}\label{thm-mainalg}
Let $H$ be a $t$-apex graph and let $b\ge t+4$ be an integer. There exists a polynomial-time algorithm that,
given a $(t+3)$-connected $H$-minor-free graph $G$ and an assignment $L$ of lists of size at least $t+4$ and at most $b$ to vertices of $G$,
decides whether $G$ is $L$-colorable.
\end{theorem}

Consequently, $k$-LC is polynomial-time decidable for $(t+3)$-connected $t$-apex-minor-free graphs for every $k\ge t+4$.
Note that for every surface $\Sigma$, there exists a $1$-apex graph that cannot be embedded in $\Sigma$.
Hence, Theorem~\ref{thm-mainalg} implies that $5$-LC is polynomial-time decidable for $4$-connected graphs embedded in $\Sigma$,
which is somewhat weaker than the previously mentioned results~\cite{dvokawalg,lukethe}.
However, the constraints on the number of colors and the connectivity cannot be relaxed in general, unless $\text{P}=\text{NP}$.
For the number of colors, we have the following.

\begin{theorem}\label{thm-hard1}
For every integer $t\ge 1$, the problem $(t+3)$-LC is NP-complete for $(t+3)$-connected $K_{t+4}$-minor-free graphs.
\end{theorem}

Let us remark that $K_{t+4}$ is a $t$-apex graph for every $t\ge 0$.  The following theorem deals with the case that
the connectivity restriction in Theorem~\ref{thm-mainalg} is relaxed.  We say that a graph $G$ is \emph{$(t,c)$-apex-free}
if every $(t+3)$-connected minor of $G$ with at least $t+c+7$ vertices is $(t-1)$-apex.

\begin{theorem}\label{thm-hard2}
For all integers $t\ge 1$ and $c\ge 2$, it is NP-complete to decide whether a $(t+2)$-connected $(t,c)$-apex-free
graph is $(t+c)$-colorable.
\end{theorem}

For instance, consider the case that $c=4$ and let $H$ be the complete join of a $3$-connected planar graph with at least
$11$ vertices with a clique on $t$ vertices (the \emph{complete join} of graphs $G_1$ and $G_2$
is the graph obtained from their disjoint union by adding all edges with one end in $V(G_1)$ and the other end in $V(G_2)$).
Then $H$ has at least $t+11$ vertices, it is $(t+3)$-connected, and it is $t$-apex but not $(t-1)$-apex.
Consequently, every $(t,4)$-apex-free graph is $H$-minor-free.
The previous theorem shows that it is NP-complete to decide whether a $(t+2)$-connected $H$-minor-free graph is $(t+4)$-colorable,
in contrast with Theorem~\ref{thm-mainalg}.

Furthermore, note that forbidding a $0$-apex (i.e., planar) graph
as a minor ensures bounded tree-width~\cite{RSey}, and thus $k$-LC is polynomial-time decidable for graphs avoiding a $0$-apex
graph as a minor, for every $k\ge 1$.

In the following section, we design the algorithm of Theorem~\ref{thm-mainalg}.
The algorithm uses a variant of the structure theorem for graphs avoiding a $t$-apex minor;
although it is well known among the graph minor research community, we are not aware of its published proof,
and give it in Appendix for completeness.
The hardness results (Theorems~\ref{thm-hard1} and \ref{thm-hard2}) are proved
in Section~\ref{sec-complex}.

\section{Algorithm}\label{sec-alg}

Let $G$ be a graph with a list assignment $L$ and let $X\subseteq V(G)$ be a set of vertices. We let $\Phi(G,L,X)$ denote the set of restrictions of $L$-colorings of $G$ to $X$.  In other words,
$\Phi(G,L,X)$ is the set of $L$-colorings of $X$ which extend to an $L$-coloring of $G$.
We say that $G$ is \emph{critical with respect to $L$} if $G$ is not $L$-colorable, but every proper subgraph of $G$ is $L$-colorable.
The graph $G$ is \emph{$X$-critical with respect to $L$} if $\Phi(G,L,X)\neq \Phi(G',L,X)$ for every proper subgraph $G'$ of $G$ such that $X\subseteq V(G')$.
Thus, removing any part of an $X$-critical graph affects which colorings of $X$ extend to the whole graph.
Postle~\cite{lukethe} gave the following bound on the size of critical graphs.  A \emph{closed disk} is a set homeomorphic to $\{(x,y):x^2+y^2\le 1\}$, and an \emph{open disk} is a set homeomorphic to $\{(x,y):x^2+y^2<1\}$.

\begin{theorem}[{Postle~\cite[Lemma 3.6.1]{lukethe}}]\label{thm-lukebnd}
Let $G$ be a graph embedded in a closed disk, let $L$ be a $(\ge\!5)$-list assignment for $G$, and
let $X$ be the set of vertices of $G$ drawn in the boundary of the disk.
If $G$ is $X$-critical with respect to $L$, then $|V(G)|\le 29|X|$.
\end{theorem}

Theorem~\ref{thm-lukebnd} has several surprising corollaries; in particular, it makes it possible to test whether a precoloring
of an arbitrary connected subgraph (of unbounded size) extends to a coloring of a graph embedded in a fixed surface from lists of size five~\cite{dvokawalg}.
Based on somewhat similar ideas, we use it here to deal with a restricted case of list-coloring graphs from a proper minor-closed class.
The algorithm uses the structure theorem of Robertson and Seymour~\cite{robertson2003graph}, which asserts that each graph from a proper
minor-closed class can be obtained by clique-sums from graphs which are ``almost'' embedded in a surface of bounded genus, up to vortices and apices.

\subsection{Embedded graphs}
As the first step, let us consider graphs that can be drawn in a fixed surface.
A fundamental consequence of Theorem~\ref{thm-lukebnd} is that actually all the vertices of $G$ are at distance $O(\log |X|)$ from $X$.

\begin{lemma}[{Postle~\cite[Theorem 3.6.3]{lukethe}}]\label{lemma-logdist}
Let $G$ be a graph embedded in a closed disk, let $L$ be a $(\ge\!5)$-list assignment for $G$, and let $X$ be the set of vertices of $G$ drawn in the boundary of the disk.
If $G$ is $X$-critical with respect to $L$, then every vertex of $G$ is at distance at most $58\log_2|X|$ from $X$.
\end{lemma}

Hence, removing vertices sufficiently distant from $X$ cannot affect which precolorings of $X$ extend.
A \emph{$k$-nest} in a graph $G$ embedded in a surface, with respect to a set $X\subseteq V(G)$,
is a set $\Delta_0\supset\Delta_1\supset\ldots\supset \Delta_k$ of closed disks in $\Sigma$ bounded by pairwise vertex-disjoint cycles of $G$
such that no vertex of $X$ is drawn in the interior of $\Delta_0$ and at least one vertex of $G$ is drawn in the interior of $\Delta_k$.
The vertices drawn in the interior of $\Delta_k$ are called \emph{eggs}.

\begin{lemma}\label{lemma-rednest}
Let $G$ be a graph embedded in a surface $\Sigma$, let $L$ be a $(\ge\!5)$-list assignment for $G$, let $X$ be a subset of $V(G)$ and let $\Delta_0\supset\Delta_1\supset\ldots\supset \Delta_k$ be 
a $k$-nest in $G$ with respect to $X$.  If $k\ge 58\log_2|V(G)|$ and $v$ is an egg, then $\Phi(G,L,X)=\Phi(G-v,L,X)$.
\end{lemma}
\begin{proof}
Clearly, $\Phi(G-v,L,X)\supseteq \Phi(G,L,X)$, hence it suffices to show that $\Phi(G-v,L,X)\subseteq \Phi(G,L,X)$.

Let $G_1$ be the subgraph of $G$ drawn in $\Delta_0$ and let $Y$ be the set of vertices of $G$ contained in the boundary of $\Delta_0$.
Let $G'_1$ be a minimal subgraph of $G_1$ such that $Y\subseteq V(G'_1)$ and $\Phi(G'_1,L,Y)=\Phi(G_1,L,Y)$.
Observe that $G'_1$ is $Y$-critical (with respect to $L$).  Since $v$ is an egg of the $k$-nest, its distance
from $X$ in $G_1$ is greater than $k$.  On the other hand, all vertices of $G'_1$ are at distance at most $k$ from $Y$
by Lemma~\ref{lemma-logdist}.  Therefore, $v\not\in V(G'_1)$.
It follows that $G'_1\subseteq G_1-v$, and thus $\Phi(G_1-v,L,Y)\subseteq \Phi(G'_1,L,Y)=\Phi(G_1,L,Y)$.

Let $\phi$ be an $L$-coloring of $G-v$ and let $\phi_0\in \Phi(G-v,L,X)$ be its restriction to $X$.
Let $\phi_1\in \Phi(G_1,L,Y)$ be the restriction of $\phi$ to $Y$.
Since $\Phi(G_1-v,L,Y)=\Phi(G_1,L,Y)$, it follows that $\phi_1$ belongs to $\Phi(G_1,L,Y)$, i.e., $\phi_1$ extends to an $L$-coloring $\psi_1$ of $G_1$.
Note that $G$ has no edges between the vertices in the interior of $\Delta_0$ and those in $\Sigma\setminus\Delta_0$.
Hence, we can obtain an $L$-coloring $\psi$ of $G$ by setting $\psi(v)=\psi_1(v)$ for $v\in V(G_1)$ and $\psi(v)=\phi(v)$ for $v\in V(G)\setminus V(G_1)$.
The restriction of $\psi$ to $X$ is equal to $\phi_0$, showing that $\phi_0\in \Phi(G,L,X)$.  As the choice of $\phi_0\in \Phi(G-v,L,X)$ was arbitrary,
it follows that $\Phi(G-v,L,X)\subseteq \Phi(G,L,X)$.
\end{proof}

Therefore, we can remove vertices within deeply nested cycles without affecting which colorings of $X$ extend.
This is sufficient to restrict tree-width, as we show below.
We use the result of Geelen at al.~\cite{GRS} regarding existence of planarly embedded subgrids in grids on surfaces.
For integers $a,b\ge 2$, an \emph{$a\times b$ grid} is the Cartesian product of a path with $a$ vertices with a path with $b$ vertices.
An embedding of a grid $G$ in a closed disk is \emph{canonical} if the outer cycle of the grid forms the boundary of the disk.

\begin{lemma}[Geelen at al.~\cite{GRS}]\label{lemma-plangrid}
Let $g\ge 0$ and $r,s\ge 2$ be integers satisfying $s\le r/\lceil\sqrt{g+1}\rceil-1$.
If $H$ is an $r\times r$ grid embedded in a surface $\Sigma$ of Euler genus $g$,
then an $s\times s$ subgrid $H'$ of $H$ is canonically embedded in a closed disk $\Delta\subseteq\Sigma$.
\end{lemma}

We also need a bound on the size of grid minors in embedded graphs of large tree-width.
Recall that a \emph{tree decomposition} of a graph $G$ consists of a tree $T$ and a function $\beta:V(T)\to 2^{V(G)}$
such that
\begin{itemize}
\item for each edge $uv\in E(G)$, there exists $x\in V(T)$ such that $\{u,v\}\subseteq \beta(x)$, and
\item for each $v\in V(G)$, the set $\{x\in V(T):v\in \beta(x)\}$ induces a non-empty connected subtree of $T$.
\end{itemize}
The sets $\beta(x)$ for $x\in V(T)$ are the \emph{bags} of the decomposition.
The \emph{width} of the tree decomposition is the maximum of the sizes of its bags minus one.
The \emph{tree-width} $\tw(G)$ of a graph $G$ is the minimum width of its tree decomposition.

\begin{lemma}[Theorem 4.12 in Demaine et al.~\cite{demainesurf}]\label{lemma-dem}
Let $r\ge 2$ and $g\ge 0$ be integers.  If $G$ is a graph embedded in a surface of Euler genus $g$
and $\tw(G)>6(g+1)r$, then $G$ contains an $r\times r$ grid as a minor.
\end{lemma}

\begin{lemma}\label{lemma-redwidth}
Let $G$ be a graph embedded in a surface $\Sigma$ of Euler genus $g$, let $X$ be a subset of $V(G)$ and let $F$ be a set of faces of $G$ such that
every vertex of $X$ is incident with a face belonging to $F$.  If $G$ contains no $k$-nest with respect to $X$,
then $\tw(G)\le 12(g+1)\lceil\sqrt{g+|F|+1}\rceil(k+2)$.
\end{lemma}
\begin{proof}
Suppose for a contradiction that $G$ contains no $k$-nest with respect to $X$ and that $\tw(G)>12(g+1)\lceil\sqrt{g+|F|+1}\rceil(k+2)$.
Let $\Sigma'$ be the surface obtained from $\Sigma$ by adding a crosscap in each of the faces of $F$.
Note that the Euler genus of $\Sigma'$ is $g'=g+|F|$.  Since $G$ contains no $k$-nest with respect to $X$,
observe that the drawing of $G$ in $\Sigma'$ contains no $k$-nest with respect to $\emptyset$.

Let $r=(2k+4)\lceil\sqrt{g'+1}\rceil$.  By Lemma~\ref{lemma-dem}, $G$ contains an $r\times r$ grid $H$ as a minor.
The embedding of $G$ in $\Sigma'$ specifies an embedding of $H$ in $\Sigma'$.  By Lemma~\ref{lemma-plangrid}, $H$ contains
a $(2k+3)\times (2k+3)$ subgrid embedded in a disk $\Delta\subseteq \Sigma'$.  However, such a subgrid contains
a $k$-nest with respect to $\emptyset$, and consequently the embedding of $G$ in $\Sigma'$ contains a $k$-nest with respect to $\emptyset$.
This is a contradiction.
\end{proof}

Let $G$ be a graph embedded in a surface $\Sigma$ and let $X$ be a subset of $V(G)$.  A graph $G'$ is a \emph{$k$-nest reduction of $G$ with respect to $X$}
if it is obtained from $G$ by repeatedly finding a $k$-nest with respect to $X$ and removing its egg, until there is no such $k$-nest.
To test whether a vertex $v$ is an egg of a $k$-nest with respect to $X$, we proceed as follows: take all faces incident with $v$.  If
their union contains a non-contractible curve, then $v$ is not an egg of a $k$-nest.  Otherwise, the union of their boundaries contains a cycle
$C_k$ bounding a disk $\Delta_k$ containing $v$.  Next, we similarly consider the union of $\Delta_k$ and all the faces incident with vertices
of $C_k$, and either conclude that $v$ is not an egg of a $k$-nest, or obtain a cycle $C_{k-1}$ bounding a disk $\Delta_{k-1}\supset \Delta_k$.  We proceed in the
same way until we obtain the disk $\Delta_0$.  Finally, we check whether the interior of $\Delta_0$ contains a vertex of $X$ or not.
This can be implemented in linear time.  By repeatedly applying this test and removing the eggs, we can obtain a $k$-nest reduction in quadratic time.

Thus, we have a simple polynomial-time algorithm for deciding colorability
of an embedded graph $G$ from lists of size $5$: find a $k$-nest reduction $G'$ of $G$, where $k$ is given by Lemma~\ref{lemma-rednest}.
By Lemma~\ref{lemma-redwidth}, the resulting graph has tree-width at most $O(\log |V(G)|)$, and thus we can test its colorability
using the standard dynamic programming approach in polynomial time (see~\cite{listcoltw} for details).

\subsection{Vortices}
Next, we deal with vortices.  A \emph{path decomposition} of a graph is its tree decomposition $(T,\beta)$
such that $T$ is a path.
A \emph{vortex} consists of a graph $F$ and a sequence $v_1$, \ldots, $v_t$ of distinct vertices of $F$,
together with a path decomposition of $F$ with bags $B_1$, \ldots, $B_t$ in order along the path such that
$v_i\in B_i$ for $i=1,\ldots, t$.  The \emph{depth} of the vortex is the order of the
largest of the bags of the decomposition.  The sequence $v_1$, \ldots, $v_t$ is called the \emph{boundary} of the vortex.
Let $\Sigma$ be a surface. A graph $G$ is \emph{almost embedded in $\Sigma$, with vortices $G_1, \ldots, G_m$},
if $G=G_0\cup G_1\cup \ldots\cup G_m$ for some graph $G_0$ such that
\begin{itemize}
\item $V(G_i)\cap V(G_j)=\emptyset$ for $1\le i<j\le m$,
\item $V(G_0)\cap V(G_i)$ is exactly the set of boundary vertices of the vortex $G_i$, for $1\le i\le m$, and
\item there exists an embedding of $G_0$ in $\Sigma$ and pairwise disjoint closed disks $\Delta_1, \ldots, \Delta_m\subset \Sigma$
such that for $1\le i\le m$, the embedding of $G_0$ intersects $\Delta_i$ exactly in the set of boundary vertices of $G_i$,
which are drawn in the boundary of $\Delta_i$ in order that matches the order prescribed by the vortex (up to reflection and circular shift).
\end{itemize}
We say that $G_0$ is the \emph{embedded part} of $G$, and the set of \emph{boundary vertices} of $G_0$ is the union of the boundary
sets of the vortices.

The tree-width of a graph with vortices depends on their depth as follows (see also~\cite[Lemma 5]{Demaine2008} for a similar bound).
\begin{lemma}\label{lemma-vortdepth}
Let $G=G_0\cup G_1\cup \ldots\cup G_m$ for some graph $G_0$ and vortices $G_1$, \ldots, $G_m$ of depth at most $d$.
Suppose that
\begin{itemize}
\item $V(G_i)\cap V(G_j)=\emptyset$ for $1\le i<j\le m$, and
\item $V(G_0)\cap V(G_i)$ is exactly the set of boundary vertices of the vortex $G_i$, for $1\le i\le m$, and
\item the boundary vertices of the vortex $G_i$ in order form a path in $G_0$, for $1\le i\le m$.
\end{itemize}
Then, $\tw(G)\le d(\tw(G_0)+1)-1$.
\end{lemma}
\begin{proof}
Consider a tree decomposition $(T,\beta)$ of $G_0$ such that each bag of this decomposition has order at most $\tw(G_0)+1$.
For each boundary vertex $v$ of a vortex, let $X_v$ be the corresponding bag in the path decomposition of the vortex.
For all other vertices, let $X_v=\{v\}$.
For each $x\in V(T)$, let $\beta'(x)=\bigcup_{v\in \beta(x)} X_v$.

Consider an edge $uv\in E(G)$.  If $uv\in E(G_0)$, then there exists $x\in V(T)$ with $\{u,v\}\subseteq \beta(x)\subseteq \beta'(x)$.
If $uv\not\in E(G_0)$, then $uv$ is an edge of one of the vortices, and thus there exists a vertex $w\in V(G_0)$ such that
$\{u,v\}\subseteq X_w$.  Since $(T,\beta)$ is a tree decomposition of $G_0$, there exists $x\in V(T)$ such that $w\in \beta(x)$, and
thus $\{u,v\}\subseteq \beta'(x)$.

Next, consider a vertex $v\in V(G)$.  If $v\in V(G_0)$, then let $Z_0=\{x\in V(T):v\in \beta(x)\}$, otherwise let $Z_0=\emptyset$.
Since $(T,\beta)$ is a tree decomposition of $G_0$, $Z_0$ induces a connected subtree of $T$.
If $v$ belongs to a vortex, say to $G_1$, then let $Y$ be the set of boundary vertices of $G_1$ whose bags in the path decomposition of $G_1$
contain $v$, and let $Z_1=\{x\in V(T):\beta(x)\cap Y\neq\emptyset\}$; otherwise, let $Z_1=\emptyset$.
The elements of $Y$ form a contiguous interval in the sequence of boundary vertices of $G_1$, and thus
they form a path in $G_0$.  Since this path is a connected subgraph of $G_0$ and $(T,\beta)$ is a tree decomposition of $G_0$,
we conclude that $Z_1$ induces a connected subtree of $T$.
Observe that at least one of $Z_0$ and $Z_1$ is non-empty, and if they are both non-empty, then they are not disjoint.
Consequently, $\{x\in V(T):v\in \beta'(x)\}=Z_0\cup Z_1$ induces a non-empty connected subtree of $T$.

It follows that $(T,\beta')$ is a tree decomposition of $G$.
Since every bag of $(T,\beta')$ has order at most $d(\tw(G_0)+1)$, the claim of the lemma follows.
\end{proof}

\subsection{Structure theorem}
A \emph{clique-sum} of two graphs $G_1$ and $G_2$ is a graph obtained from them by choosing cliques
of the same size in $G_1$ and $G_2$,
identifying the two cliques, and possibly removing some edges of the resulting clique.
The usual form of the structure theorem for graphs avoiding a fixed minor is as follows~\cite{robertson2003graph}.
\begin{theorem}\label{thm-strbasic}
For any graph $H$, there exist integers $m, d, a\ge 0$ with the following property.
If $G$ is $H$-minor-free, then
$G$ is a clique-sum of graphs $G_1$, \ldots, $G_s$ such that for $i=1,\ldots,s$,
there exists a surface $\Sigma_i$ and a set $A_i\subseteq V(G_i)$ satisfying the following:
\begin{itemize}
\item[\textrm{(a)}] $H$ cannot be drawn in $\Sigma_i$,
\item[\textrm{(b)}] $|A_i|\le a$, and
\item[\textrm{(c)}] $G_i-A_i$ can be almost embedded in $\Sigma_i$ with at most $m$ vortices of depth at most $d$.
\end{itemize}
\end{theorem}

The graphs $G_1$, \ldots, $G_s$ are called the \emph{pieces} of the decomposition, and for each piece $G_i$,
we say that the vertices of $A_i$ are its \emph{apices}.
If the embedded part of $G_i-A_i$ has at most four vertices, then we say that the piece $G_i$ is \emph{degenerate}.
Let us remark that it is possible that $\Sigma_i$ is null for some $i\in\{1,\ldots, s\}$, and thus $A_i=V(G_i)$,
and in particular if $H$ is planar then Theorem~\ref{thm-strbasic} claims that every $H$-minor-free graph $G$
is a clique-sum of graphs of bounded size.

We need a strengthening of Theorem~\ref{thm-strbasic} that restricts the apex vertices
as well as the properties of the embedding.  Given a graph $G$ and a subset $A$ of
its vertices such that $G-A$ is almost embedded in some surface with vortices $G_1$, \ldots, $G_m$,
we say that a vertex $v\in A$ is a \emph{major apex vertex} if $v$ has some neighbor in $G$ not belonging
to $A\cup V(G_1\cup \ldots G_m)$, i.e., if $v$ has a neighbor in the surface part of the almost-embedding.
For a graph $H$ and a surface $\Sigma$, let $a(H,\Sigma)$ denote the smallest size
of a subset $B$ of vertices of $H$ such that $H-B$ can be embedded in $\Sigma$.
A face is \emph{2-cell} if it is homeomorphic to an open disk.

\begin{theorem}\label{thm-strapex}
For any graph $H$, there exist integers $m$, $d$ and $a$ with the following property.
If $G$ is $H$-minor-free, then
$G$ is a clique-sum of graphs $G_1$, \ldots, $G_s$ such that for $i=1,\ldots, s$,
there exists a surface $\Sigma_i$ and a set $A_i\subseteq V(G_i)$ satisfying the following:
\begin{itemize}
\item[\textrm{(a)}] $H$ cannot be drawn in $\Sigma_i$,
\item[\textrm{(b)}] $|A_i|\le a$,
\item[\textrm{(c)}] $G_i-A_i$ can be almost embedded in $\Sigma_i$ with at most $m$ vortices of depth at most $d$,
\item[\textrm{(d)}] every triangle in the embedded part bounds a $2$-cell face, and
\item[\textrm{(e)}] at most $a(H,\Sigma_i)-1$ apices of $A_i$ are major.
\end{itemize}
\end{theorem}

The condition (d) is just a simple technicality.
That it is possible to restrict the major apices along the lines of condition (e)
is well known among the graph minor research community, but as far as we are aware,
it has never been published in this form.  For this reason, we provide a proof in Appendix.  Let us also remark
that the decomposition of Theorem~\ref{thm-strapex} can be found in polynomial time in the same way as the decomposition of Theorem~\ref{thm-strbasic} (see~\cite{minalg1,minalg2} for details),
as all the steps of the proof outlined in Appendix can be carried out in polynomial time.

Suppose that $H$ is a $t$-apex graph and that $G$ is a $(t+3)$-connected $H$-minor-free graph $G$.
For $1\le i\le s$, let $G_i$, $A_i$ and $\Sigma_i$ be as in Theorem~\ref{thm-strapex},
and let $G'_i$ be the embedded part of $G_i-A_i$.
Let $A'_i\subseteq A_i$ be the set of the major apices; we have $|A'_i|\le a(H,\Sigma_i)-1\le t-1$.
Suppose that $G_i$ is not degenerate and that $K$ is a clique in $G_i$ through that $G_i$ is summed with other pieces of the decomposition,
and that $K$ contains a vertex that is neither in $A_i$ nor in the vortices.  It follows that $V(K)\subseteq V(G'_i)\cup A'_i$.
Note that $V(K)$ forms a cut in $G$, and thus $|V(K)|\ge t+3$.  Therefore, the subclique $K'=K-A'_i$
has size at least four.  Since $K'\subseteq G'_i$
and every triangle in $G'_i$ bounds a $2$-cell face, it follows that
$|V(K')|=4$ and that $G'_i=K'$.  However, this contradicts the assumption that $G_i$ is not degenerate.

Therefore, we conclude that for each non-degenerate piece $G_i$, all the clique-sums are over
cliques whose vertices are contained in the union of $A_i$ and the vortices.

\begin{lemma}\label{lemma-rwg}
Let $H$ be a $t$-apex graph, let $G$ be a $(t+3)$-connected $H$-minor-free graph and let $L$ be a $(\ge t+4)$-list assignment for $G$.
Let $G_1$, $\ldots, G_s$ be the pieces of a decomposition of $G$ as in Theorem~\ref{thm-strapex}.
For $1\le i\le s$, if $G_i$ is degenerate, then let $G'_i$ and $G''_i$ be null.  Otherwise, let $G'_i$ be the embedded part of $G_i-A_i$
and let $X_i$ be the set of its boundary vertices.  Let $G''_i$ be a $k$-nest reduction of $G'_i$ with respect to $X_i$, where
$k=\lceil 58\log_2 |V(G)|\rceil+1$.  Let $G'=G-\bigcup_{i=1}^s (V(G'_i)\setminus V(G''_i))$.
Then $G$ is $L$-colorable if and only if $G'$ is $L$-colorable.
\end{lemma}
\begin{proof}
For $1\le i\le s$, let $A_i$ and $\Sigma_i$ be as in Theorem~\ref{thm-strapex}.
Let $A'_i\subseteq A_i$ be the set of major apices; recall that $|A'_i|\le t-1$.
Let $F$ be the graph obtained from the embedded part $G'_i$ of $G_i-A_i$ by removing the edges between vertices of $X_i$.
As we observed, if $G_i$ is not degenerate, then all the clique-sums in the decomposition of $G$ involving $G_i$ are over
cliques whose vertices are contained in the union of $A_i$ and the vertex sets of the vortices of $G_i$.  Hence,
$F$ is a subgraph of $G$.
Let $F^\star$ be the subgraph of $G$ consisting of $F$, the major apices $A'_i$ and all edges between $G'_i$ and $A'_i$.

Suppose that $v$ is an egg of a $k$-nest in $G'_i$ with respect to $X_i$; then $v$ is an egg of a $(k-1)$-nest in $F$
with respect to $X_i$.
Consider any $L$-coloring $\psi$ of the vertices of $A'_i$, and let $L'$ be the list assignment for $F$ defined by
$L'(w)=L(w)\setminus \{\psi(u):u\in A'_i, uw\in E(G)\}$.  Note that $|L'(w)|\ge (t+4)-(t-1)=5$.
By Lemma~\ref{lemma-rednest}, we have $\Phi(F,L',X_i)=\Phi(F-v,L',X_i)$.
Since this holds for every $\psi$, we conclude that $\Phi(F^\star, L,X_i\cup A'_i)=\Phi(F^\star-v, L,X_i\cup A'_i)$.
Since $X_i\cup A'_i$ separates $F^\star-(X_i\cup A'_i)$ from the rest of $G$, it follows that removing $v$ does not affect the $L$-colorability of $G$.
Repeating this idea for all removed eggs in all the pieces,
we conclude that $G$ is $L$-colorable if and only if $G'$ is $L$-colorable.
\end{proof}

\subsection{The algorithm}

\begin{proof}[Proof of Theorem~\ref{thm-mainalg}]
Let $k=\lceil 58\log_2 |V(G)|\rceil+1$.
For $1\le i\le s$, let $G_i$, $A_i$ and $\Sigma_i$ be as in Theorem~\ref{thm-strapex}.
In polynomial time, we can find a reduction $G'$ of $G$ as in Lemma~\ref{lemma-rwg}.

Let us consider some $i\in\{1,\ldots, s\}$, and let $G'_i=G_i-(V(G)\setminus V(G'))$.
Note that the graph $G'_i-A_i$ is almost embedded with at most $m$ vortices of depth at most $d$ in $\Sigma_i$, and
the embedded part has no $k$-nest with respect to the boundaries of the vortices (this is obvious if $G_i$
is degenerate and follows by the construction of $G'$ otherwise).
Let $G''_i$ be obtained from $G'_i-A_i$ by adding edges that trace the boundaries of all the vortices in $\Sigma_i$.  Note that
the embedded part $G'''_i$ of $G''_i$ has no $(k+1)$-nest with respect to the boundaries of the vortices.
By Lemma~\ref{lemma-redwidth}, we have $\tw(G'''_i)\le 12(g+1)\lceil\sqrt{g+m+1}\rceil(k+3)$, where $g$ is the Euler genus of $H$.
By Lemma~\ref{lemma-vortdepth}, we have $\tw(G''_i)\le d(12(g+1)\lceil\sqrt{g+m+1}\rceil(k+3)+1)-1$.
Note that $\tw(G'_i)\le a+\tw(G''_i)$, and thus $\tw(G'_i)=O(\log |V(G)|)$.

The graph $G'$ is a clique-sum of $G'_1$, $G'_2$, \ldots, $G'_s$, and thus $\tw(G')=O(\log |V(G)|)$.
Since the sizes of the lists are bounded by the constant $b$,
the algorithm of Jansen and Scheffler~\cite{listcoltw} enables us to test $L$-colorablity of $G'$ in time $2^{O(\tw(G'))}|V(G')|$, which
is polynomial in $|V(G)|$.  By Lemma~\ref{lemma-rwg}, $G$ is $L$-colorable if and only if $G'$ is $L$-colorable.
\end{proof}

Let us remark that without the upper bound $b$ on the sizes of the lists, we would only get an algorithm with time complexity
$|V(G)|^{O(\log |V(G)|)}$.

\section{Complexity}\label{sec-complex}

We proceed with the NP-hardness arguments.  Let us start with a simple observation.

\begin{lemma}\label{lemma-adda}
Let $G_1$ and $G_2$ be graphs, let $a\ge 1$ be an integer and let $G'_1$ be the complete join of $G_1$ and $K_a$.
Then
\begin{itemize}
\item $G_2$ is a minor of $G'_1$ if and only if there exists a set $X\subseteq V(G_2)$ of size at most $a$
such that $G_2-X$ is a minor of $G_1$,
\item for every $t\ge 0$, $G_1$ is $t$-apex if and only if $G'_1$ is $(t+a)$-apex, and
\item for every $k\ge 1$, $G_1$ is $k$-connected if and only if $G'_1$ is $(k+a)$-connected.
\end{itemize}
\end{lemma}

Let us now prove the hardness results justifying the choice of the assumptions in Theorem~\ref{thm-mainalg}.

\begin{proof}[Proof of Theorem~\ref{thm-hard1}]
By Lemma~\ref{lemma-adda}, it suffices to consider the case $t=1$, i.e., to show that 4-LC is NP-complete for $K_5$-minor-free $4$-connected
graphs.  Note that a $4$-connected graph is $K_5$-minor-free if and only if it is planar~\cite{wagner},
and that 4-LC is known to be NP-complete for planar graphs~\cite{chooscomplex}, and thus we only need to deal with the connectivity restriction.
We find a reduction from $3$-colorability of connected planar graphs, which is a well-known NP-complete problem~\cite{garey1979computers}.

\begin{figure}
\begin{center}
\includegraphics{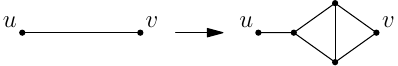}
\end{center}
\caption{A quasiedge replacement.}\label{fig-qe}
\end{figure}

Let $G$ be a connected planar graph.  Let $G_1$ be obtained from $G$ by replacing each edge $uv$ by a subgraph depicted in Figure~\ref{fig-qe}.
Observe that $G$ is $3$-colorable if and only if $G_1$ is $3$-colorable.  Furthermore, $G_1$ does not contain separating triangles,
and every vertex of $G_1$ is incident with a face of length greater than three.

Gutner~\cite{chooscomplex} constructed a plane graph $H$ without separating triangles that is
critical with respect to a $4$-list assignment $L$.  We can assume that $L$ does not use colors $1$, $2$ and $3$.
Let $x$ be an arbitrary vertex incident with the outer face of $H$ and let $L'$ be the list assignment obtained from $L$ by removing any three colors from the list of $x$
and adding colors $1$, $2$ and $3$ instead.  Since $H$ is critical with respect to $L$, it follows that every $L'$-coloring $\psi$ of $H$
satisfies $\psi(x)\in\{1,2,3\}$, and furthermore for every $i\in\{1,2,3\}$, there exists an $L'$-coloring $\psi_i$ of $H$ such that $\psi_i(x)=i$.

Let $G_2$ be the graph obtained from $G_1$ as follows. For each vertex $v\in V(G_1)$, add a copy $H_v$ of $H$ and identify its vertex $x$ with $v$.
The graph $H_v$ is drawn in the face of $G_1$ incident with $v$ of length at least four, so that $G_2$ has no separating triangles.
Let $L_2$ be the list assignment for $G$ obtained as the union of the list assignments $L'$ for the copies of $H$ appearing in $G_2$.
Note that $G_2$ is $L_2$-colorable if and only if $G_1$ is $3$-colorable.

Finally, let $G_3$ be obtained from $G_2$ as follows.  For each face $f$ of $G_2$ of length at least $4$, consider its boundary walk $v_1v_2\ldots v_m$.
Add to $f$ a wheel with rim $w_1w_2\ldots w_m$ and add edges $v_iw_i$ and $v_iw_{i+1}$ for $1\le i\le m$, where $w_{m+1}=w_1$.
Let $L_3$ be the list assignment obtained from $L_2$ by giving each vertex of the newly added wheels a list of size four disjoint
from the lists of all other vertices of $G_3$.  Clearly, $G_3$ is $L_3$-colorable if and only if $G_2$ is $L_2$-colorable.  Furthermore, $G_3$ is
a triangulation without separating triangles, and thus it is $4$-connected.

This gives a polynomial-time algorithm that, given a connected planar graph $G$, constructs a $4$-connected planar graph $G_3$ and a $4$-list assignment $L_3$
such that $G$ is $3$-colorable if and only if $G_3$ is $L_3$-colorable.  Therefore, $4$-list colorability of $4$-connected planar graphs is NP-complete.
\end{proof}

More interestingly, let us argue that the connectivity assumption is necessary, even if we consider ordinary coloring instead
of list coloring.  Let $k\ge 3$ be an integer, let $G$ be a graph, let $X$ be a triple of vertices of $G$ and let $S$ be a set of $k$-colorings of $X$ closed under permutations of colors.
Let $L$ be the list assignment such that $L(v)=\{1,\ldots, k\}$ for every $v\in V(G)$.  If $\Phi(G,L,X)=S$, then we say
that $G$ is an \emph{$S$-gadget on $X$ for $k$-coloring}.  We say that the gadget $G$ is \emph{internally $3$-connected}
if $G$ is $2$-connected and there exists no set $Z\subseteq V(G)$ of size at most $2$ such that $G-Z$ is disconnected and
all vertices of $X\setminus Z$ are contained in one component of $G-Z$.

\begin{figure}
\center{\includegraphics{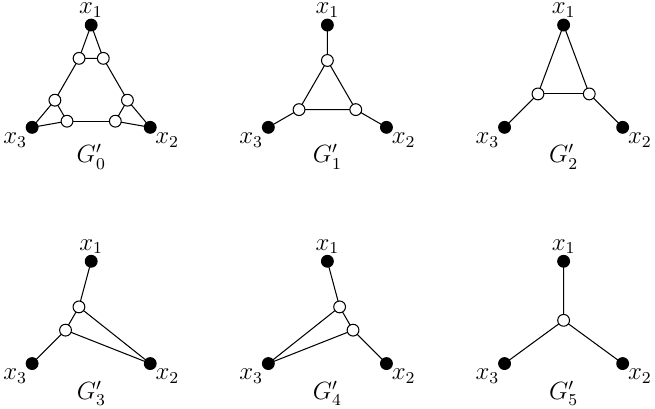}}
\caption{Gadgets from Lemma~\ref{lemma-gadget}}\label{fig-gadget}
\end{figure}

\begin{lemma}\label{lemma-gadget}
For every integer $k\ge 3$ and every set $S\subseteq \{1,\ldots,k\}^3$ closed under permutations of colors,
there exists an internally $3$-connected $S$-gadget for $k$-coloring.
\end{lemma}
\begin{proof}
Let $S_0=\{(i,j,m):1\le i,j,m\le k\}$, $S_1=S_0\setminus \{(i,i,i):1\le i\le k\}$, $S_2=S_0\setminus\{(i,j,j):1\le i,j\le k, i\neq j\}$, $S_3=S_0\setminus\{(j,i,j):1\le i,j\le k, i\neq j\}$,
$S_4=S_0\setminus\{(j,j,i):1\le i,j\le k, i\neq j\}$ and $S_5=S_0\setminus\{(i,j,m):1\le i,j,m\le k, i\neq j\neq m\neq i\}$.
For $0\le i\le 5$, let $G_i$ be the complete join of the graph $G'_i$ depicted in Figure~\ref{fig-gadget} with a clique on $k-3$ vertices.
Let $X=(x_1,x_2,x_3)$, and observe that for $0\le i\le 5$, the graph $G_i$ is a connected $S_i$-gadget on $X$ for $k$-coloring.

If $S=S_0$, then let $G=G_0$.
If $S\neq S_0$, then there exists non-empty $I_0\subseteq \{1,\ldots, 5\}$ such that $S=\bigcap_{i\in I_0} S_i$.
If $|I_0|\ge 2$, then let $I=I_0$; otherwise, let $I$ be a multiset obtained from $I_0$ by changing the multiplicity
of its element to $2$.  Let $G$ be the complete join of $\bigcup_{i\in I} G'_i$ with a clique on $k-3$ vertices.
Observe that in both cases, $G$ is an internally $3$-connected $S$-gadget for $k$-coloring.
\end{proof}

\begin{proof}[Proof of Theorem~\ref{thm-hard2}]
By Lemma~\ref{lemma-adda}, it suffices to show this claim for $t=1$, i.e., 
we need to show that $(c+1)$-colorability of $3$-connected $(1,c)$-apex-free graphs is NP-complete.
Recall that a graph is $(1,c)$-apex-free if all its $4$-connected minors with at least $c+8$ vertices are planar.
We give a reduction from planar SAT, which is known to be NP-complete~\cite{plansat}.

Let $X=(x_1,x_2,x_3)$, and let $S_0, S_1$, \ldots, $S_5$ be the set from the proof of Lemma~\ref{lemma-gadget} with $k=c+1$.
Let $A=S_2\cap S_3\cap S_4$ be the set of $(c+1)$-colorings of $X$ such that either all vertices have the same color, or they have three different colors.
Let $B=S_1\cap S_2\cap S_5$ be the set of $(c+1)$-colorings of $X$ such that the color of $x_1$ is different from the color of $x_2$ if and only if $x_1$ and $x_3$ have the same color.
Let $\Delta_A$, $\Delta_B$, and $\Delta_0$ be $3$-connected $A$-, $B$-, and $S_0$-gadgets, respectively, on $X$ for $(c+1)$-coloring,
which exist by Lemma~\ref{lemma-gadget}.  Let $\Delta'_A$, $\Delta'_B$, and $\Delta'_0$ be the graphs obtained from $\Delta_A$, $\Delta_B$, and $\Delta_0$,
respectively, by adding the edges of the triangle $x_1x_2x_3$.

Let $E$ be the graph consisting of vertices $x_1$, $x_2$, $x_3$, $x_4$, $p_1$, $q_1$, $p_2$, $q_2$ and copies of the $B$-gadget $\Delta_B$ on
$(x_1,x_2,q_1)$, $(q_1,p_1,x_1)$, $(p_1,q_1,q_2)$, $(q_2,p_2,p_1)$, $(p_2,q_2,x_4)$, and $(x_4,x_3, p_2)$,
see Figure~\ref{fig-E}(a).  Note that in any $(c+1)$-coloring of $E$, the vertices $x_1$ and $x_2$ have the same color if and only if the vertices
$x_3$ and $x_4$ have the same color.  Furthermore, any $(c+1)$-coloring of $\{x_1,x_2,x_3,x_4\}$ satisfying this condition
extends to a $(c+1)$-coloring of $E$.  We say that $E$ is the \emph{copy gadget on $(x_1,\ldots, x_4)$}.

\begin{figure}
\center{\includegraphics{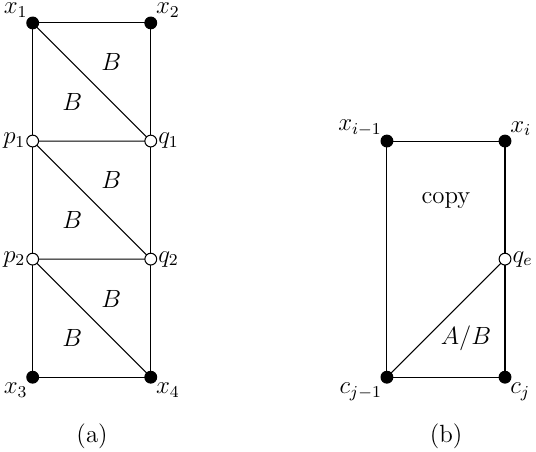}}
\caption{The copy and edge gadgets.}\label{fig-E}
\end{figure}

Given a planar instance $\phi$ of SAT, we construct a $3$-connected graph $G_{\phi}$ which is $(c+1)$-colorable if and only if the instance is satisfiable, in the following
way (see Figure~\ref{fig-exhard} for an illustration).  Let $Z_\phi$ be the incidence graph of $\phi$ drawn in plane.
By modifying the formula $\phi$ if necessary (enlarging clauses by including a single variable several times, so that
$Z_\phi$ contains parallel edges), we can assume that $Z_\phi$ is $2$-edge-connected.

\begin{figure}
\center{\includegraphics[width=130mm]{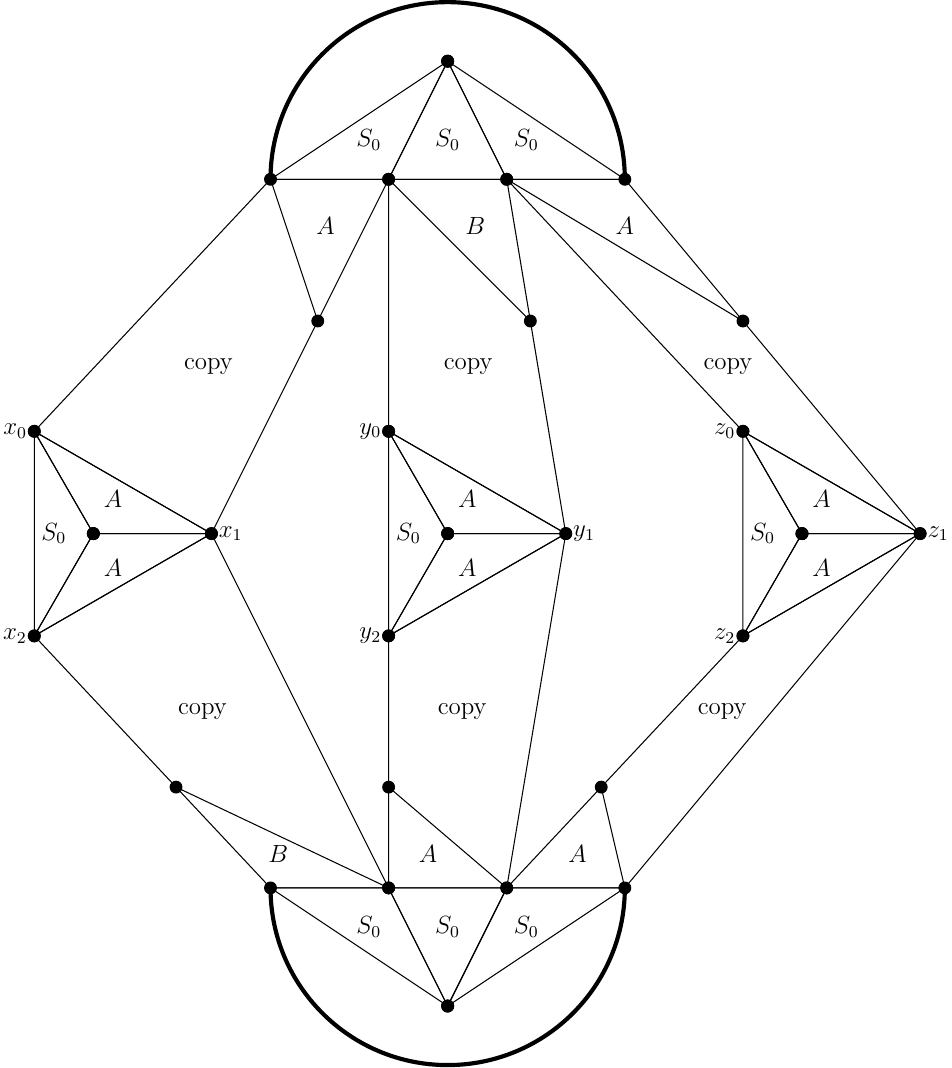}}
\caption{The graph $G_\phi$ for the formula $\phi=(x\lor \lnot y\lor z)\land (\lnot x\lor y\lor z)$.}\label{fig-exhard}
\end{figure}

First, for each variable $x$ that appears in $k$ clauses of $\phi$, we add to $G_{\phi}$ the graph $G_x$ consisting of vertices $c_x$, $x_0$, $x_1$, \ldots, $x_k$, and $k$ copies of the $A$-gadget $\Delta_A$ on
$(c_x,x_0,x_1)$, $(c_x,x_1,x_2)$, \ldots, $(c_x,x_{k-1},x_k)$, and a copy of the $S_0$-gadget $\Delta_0$ on $(c_x,x_k,x_0)$.
Note that in every $(c+1)$-coloring $\psi$ of $G_x$, either $\psi(x_{i-1})=\psi(x_i)$ for $1\le i\le k$, or $\psi(x_{i-1})\neq\psi(x_i)$ for $1\le i\le k$.
Furthermore, if $\psi$ is a $(c+1)$-coloring of $\{x_0, x_1, \ldots, x_k\}$ satisfying one of the conditions and additionally $\psi(x_i)\neq c+1$ for $0\le i\le k$,
then $\psi$ extends to a $(c+1)$-coloring of $G_x$.  Let $e_1$, \ldots, $e_k$ be the edges of $Z_\phi$ incident with $x$ in clockwise order
according to the drawing of $Z_\phi$; for $i=1,\ldots, k$, we define $\text{begin}(e_i)=(x_{i-1},x_i)$.

Next, for each clause $c$ of $\phi$ which is a conjunction of $k$ variables or their negations, we add
vertices $c_0$, $c_1$, \ldots, $c_k$ and the edge $c_0c_k$.  Furthermore, we add a vertex $m_c$
and copies of the $S_0$-gadget $\Delta_0$ on $(c_0,c_1,m_c)$, \ldots, $(c_{k-1},c_k,m_c)$.
Note that in any $(c+1)$-coloring,
at least one of the pairs $(c_0,c_1)$, $(c_1,c_2)$, \ldots, $(c_{k-1},c_k)$ receives two different colors.
Let $e_1$, \ldots, $e_k$ be the edges of $Z_\phi$ incident with $c$ in clockwise order
according to the drawing of $Z_\phi$; for $i=1,\ldots, k$, we define $\text{end}(e_i)=(c_{i-1},c_i)$.

Finally, for each edge $e$ of $Z_\phi$ with $\text{begin}(e)=(x_{i-1},x_i)$ and $\text{end}(e)=(c_{j-1},c_j)$,
add a vertex $q_e$ and the copy gadget on $(x_{i-1},x_i,c_{j-1},q_e)$.  Additionally, if the appearance of $x$ in $c$ is negated,
then add a copy of the $B$-gadget $\Delta_B$ on $(c_{j-1},q_e,c_j)$, otherwise add a copy of the $A$-gadget $\Delta_A$ on $(c_{j-1},q_e,c_j)$.
See Figure~\ref{fig-E}(b) for an illustration.

Let $G_\phi$ be the resulting graph, whose construction can clearly be performed in polynomial time.
Observe that in every $(c+1)$-coloring $\psi$ of $G_\phi$, each variable $x$ with $k$ occurrences satisfies either
$\psi(c_x)=\psi(x_0)=\psi(x_1)=\ldots=\psi(x_k)$, or $\psi(x_0)\neq\psi(x_1)$ and $\psi(x_1)\neq\psi(x_2)$ and $\ldots$ and $\psi(x_{k-1})\neq\psi(x_k)$.
Set $x$ to be false in the former case and to be true in the latter case.
For each clause $c$, the edge $c_0c_k$ implies that there exists $j$ with $\psi(c_{j-1})\neq \psi(c_j)$,
which ensures that the $j$-th literal of the clause is true in the described assignment.
Therefore, if $G_\phi$ is $(c+1)$-colorable, then $\phi$ is satisfiable.

Conversely, given a satisfying assignment to $\phi$, we can find a $(c+1)$-coloring $\psi$ of $G_\phi$.
First, set $\psi(c_x)=\psi(x_j)=1$ for each false variable $x$ with $k$ appearances and $0\le j\le k$, and $\psi(x_j)=1+(j\bmod 2)$ and $\psi(c_x)=3$
for each true variable $x$ and $0\le j\le k$.  For each clause $c$ with $k$ literals, set $\psi(c_0)=1$ and
for $1\le j\le k$ choose $\psi(c_j)\in \{1,2,3\}$ same as or different from $\psi(c_{j-1})$ depending
on whether the corresponding literal is false or true, using only colors $1$ and $2$ when $c$ has an odd
number of true literals and using color $3$ once when $c$ has an even number of true literals so that $\psi(c_0)\neq \psi(c_k)$.
The coloring can be extended to the rest of $G_\phi$ in the obvious way.
Therefore, $\phi$ is satisfiable if and only if $G_\phi$ is $(c+1)$-colorable.

By the planarity of $\phi$, observe that $G_\phi$ is obtained from a plane graph by clique-sums with copies of $\Delta'_A$, $\Delta'_B$, and $\Delta'_0$ on triangles.
Hence, every non-planar $4$-connected minor of $G$ has at most $\max(|V(\Delta_A)|,|V(\Delta_B)|,|V(\Delta_0)|)=c+7$ vertices.
Furthermore, observe that since all the gadgets are internally $3$-connected and the graph $Z_\phi$ is $2$-edge-connected, the
graph $G_\phi$ is $3$-connected.
The claim of Theorem~\ref{thm-hard2} follows.
\end{proof}

\section*{Acknowledgments}

We would like to thank Luke Postle for fruitful discussions regarding the problem.

\section*{Appendix}

Our goal in this section is to prove Theorem~\ref{thm-strapex}.  As we mentioned, the condition (d) is a simple
technicality that appeared before in the literature (with a slightly different formulation already in~\cite{rs17}).
The idea behind the condition (e) is also conceptually simple: suppose that $H-X$ can be embedded in $\Sigma_i$, where $|X|=a(H,\Sigma_i)$.  If at least
$|X|$ of the apices of $A_i$ attached to the embedded part of $G_i-A_i$ ``all over the place'', we could find a minor of $H$ in $G$ with these
apices playing the role of $X$.  So all but at most $a(H,\Sigma_i)-1$ of the apices only attach to a bounded number
of areas of small radius in the surface, and by creating additional vortices from these areas, we ensure that there
are at most $a(H,\Sigma_i)-1$ major apices.

Executing this idea formally is a straightforward, although rather lengthy, application of the tools of graph minors theory.
However, its presentation poses a conundrum: anyone familiar with the details of the graph minors theory as presented
in the series of papers by Robertson and Seymour (\cite{rs1}--\cite{rs23}) can likely devise the proof on their own,
while for anyone else the argument will be hard to follow due to usage of a large number of unfamiliar definitions
and results.

To alleviate this issue, we try to re-introduce and motivate all the important concepts and results used in the argument.
An exception is a short Subsection~\ref{sssec-clear} (marked with ($\star$) below), whose full explanation would require
introducing a number of very technical definitions not used anywhere else in the appendix.
While the rest of the appendix should be understandable to anyone with standard knowledge of graph theory, this subsection
assumes that the reader is familiar with the Robertson-Seymour series of papers on graph minor theory, and in particular
we will in that subsection refer to additional definitions and results from the series,
especially from~\cite{rs10, rs11, rs12, robertson2003graph,rs17}, not repeated here.

\subsection{Tangles}

The statement of Theorem~\ref{thm-strbasic} can be easily strengthened so that each piece $G_i$ of the decomposition of a $H$-minor-free
graph $G$ either has bounded size (and thus it is possible to set $A_i=V(G_i)$), or its embedded part $G'_i$ is ``large'', i.e.,
is embedded in $\Sigma_i$ with high representativity (unless $\Sigma_i$ is the sphere, in which case
the embedded part is guaranteed to contain a large grid minor, instead).  We will call the pieces of the latter kind
\emph{important}. While the pieces of the former kind are
somewhat arbitrary and there is no canonical way of choosing them, there is much less freedom in the choice of the important
pieces.  They are not quite uniquely determined (e.g., our strategy is based on the fact that different parts of the
piece can be included in vortices), but it turns out that in some sense the same important pieces must be represented
in every decomposition with properties described by Theorem~\ref{thm-strbasic}.

In both the proof and the applications of Theorem~\ref{thm-strbasic}, it is useful to be able to focus on just one important
piece of the decomposition in isolation, of course bearing in mind the fuzziness of their choice.  How to specify such a piece?
Consider any small set $S\subseteq V(G)$, and an important piece $G_i$ whose embedded part $G'_i$ is drawn in a surface
different from the sphere.  As $G'_i$ has representativity greater than $|S|$, exactly one component of $G-S$ contains a non-planar part of $G'_i$
(for important pieces whose embedded parts are drawn in the sphere, we can similarly determine such a unique component based on the
grid minor in $G'_i$).  Furthermore, since $G$ is a clique-sum of the pieces of the decomposition and the pieces have bounded clique number
(as can be easily seen), each two important pieces are separated by a small vertex cut in $G$.

Hence, it is natural to specify an important piece by pointing out for each small cut such a unique component that ``contains most of the
piece'' in the sense described in the previous paragraph.  More formally, a \emph{separation} of a graph $G$ is a pair $(A,B)$ such that $A$
and $B$ are edge-disjoint subgraphs of $G$ and $G=A\cup B$, and the \emph{order} of the separation is $|V(A\cap B)|$.  For an
integer $\theta\ge 1$, an \emph{orientation of $(<\!\theta)$-separations} is a set $\OO$ of separations of $G$ of order less than $\theta$
that for each separation $(A,B)$ of $G$ of order less than $\theta$ contains exactly one of $(A,B)$ and $(B,A)$.  For $(A,B)\in\OO$,
we will always interpret $B$ as the ``large'' part of the separation, e.g., in the sense of the previous paragraph.

Of course, not all orientations of $(<\!\theta)$-separations identify important pieces of the decomposition.  A bit surprisingly,
it is possible to fix this by adding two simple restrictions.  A \emph{tangle of order $\theta$} in a graph $G$ is an orientation of $(<\!\theta)$-separations $\TT$
satisfying the following tangle axioms:
\begin{itemize}
\item[(T1)] If $(A_1,B_1),(A_2,B_2),(A_3,B_3)\in\TT$, then $A_1\cup A_2\cup A_3\neq G$.
\item[(T2)] If $(A,B)\in\TT$, then $V(A)\neq V(G)$.
\end{itemize}
Both of these conditions are natural (saying that the whole graph cannot be a union of just a few small parts).
It turns out that tangles uniquely correspond to well-linked parts of $G$ that give rise to important pieces
of the structure theorem decomposition.  This should not be immediately obvious, but it is beyond the scope of
this paper to explain why this is the case; we invite a reader unfamiliar with the concept of tangles to read~\cite{rs10}
or another introductory text at this point, as their solid understanding will make reading the rest of Appendix much easier.
Let us remark that for $\theta\ge 2$, the membership of a separation $(A,B)$ in the tangle depends only on $E(A)$ (or equivalently on $E(B)=E(G)\setminus E(A))$);
in particular, the following stronger version of (T2) holds.
\begin{lemma}[A special case of {Roberson and Seymour~\cite[(2.3)]{rs10}}]\label{lemma-T2e}
Let $\TT$ be a tangle of order at least $2$ in a graph $G$.  If $(A,B)\in \TT$, then $E(A)\neq E(G)$.
\end{lemma}

We will need the following construction of tangles based on \emph{unbreakable sets}---informally, a set is unbreakable if
for every small separation, most of the set is contained in one of the parts of the separation.  More precisely,
let $Y\subseteq V(G)$ be a set of size $3\theta-2$.  We say that $Y$ is \emph{$\theta$-unbreakable}
if there exists no separation $(A,B)$ of $G$ of order less than $\theta$ such
that $|V(A)\cap V(B)|+\max(|Y\setminus V(A)|,|Y\setminus V(B)|)<3\theta-2$.
Note that if $|V(A)\cap Y|\ge \theta$, then
$|V(A)\cap V(B)|+|Y\setminus V(A)|=|V(A)\cap V(B)|+|Y|-|V(A)\cap Y|<|Y|=3\theta-2$.
Symmetrically, if $|V(B)\cap Y|\ge \theta$, then $|V(A)\cap V(B)|+|Y\setminus V(B)|<3\theta-2$.
Since $Y$ is $\theta$-unbreakable, we conclude that either $|V(A)\cap Y|\le\theta-1$ or $|V(B)\cap Y|\le\theta-1$
holds for every separation $(A,B)$ of order less than $\theta$.
Actually, exactly one of the inequalities holds, since $|Y|>2\theta-2$.
Hence the set $\TT$ consisting of the separations $(A,B)$ of $G$ of order less than $\theta$ such that $|V(A)\cap Y|\le\theta-1$
is an orientation of $(<\!\theta)$-separations.  In fact, it is even a tangle of order $\theta$. We include the simple argument,
but this was shown already in~\cite[(11.2)]{rs10}.

\begin{lemma}\label{lemma-unbreaktangle}
Let $\theta\ge 1$ be an integer, let $G$ be a graph, and let $Y\subseteq V(G)$ be a set of size $3\theta-2$.
Let $\TT$ be the set of separations $(A,B)$ of $G$ of order less than $\theta$ such that $|V(A)\cap Y|\le\theta-1$.
If $Y$ is $\theta$-unbreakable, then $\TT$ is a tangle of order $\theta$ in $G$.
\end{lemma}
\begin{proof}
We already argued that $\TT$ is an orientation of $(<\!\theta)$-separations, and thus it suffices to show that it
satisfies (T1) and (T2).  Given $(A_1,B_1)$, $(A_2,B_2)$, $(A_3,B_3)\in\TT$, we have
$|V(A_1\cup A_2\cup A_3)\cap Y|\le \sum_{i=1}^3 |V(A_i)\cap Y|\le 3\theta-3<|Y|$, and thus $A_1\cup A_2\cup A_3$ does not
contain all vertices of $Y$, and $A_1\cup A_2\cup A_3\neq G$.  Hence, (T1) holds.  Similarly, if $(A,B)\in \TT$, then
$|A\cap Y|<|Y|$, and thus $V(A)\neq V(G)$ and (T2) holds.
\end{proof}

We will also need the following construction to obtain new tangles from old ones.
Let $\TT$ be a tangle of order $\theta$ in a graph $G$, and let $Z\subseteq V(G)$ be a set of size less than $\theta$.
We define $\TT-Z$ as the set of separations $(A,B)$ of $G-Z$ of order less than $\theta-|Z|$ such that
$A=A'-Z$ and $B=B'-Z$ for some $(A',B')\in \TT$ with $Z\subseteq V(A\cap B)$.
Observe that $\TT-Z$ is a tangle in $G-Z$ of order $\theta-|Z|$ (see also~\cite[(8.5)]{rs10}).

\subsection{Local form of the structure theorem}\label{subs-loc}

Let $G_i$ be an important piece of the decomposition of an $H$-minor-free graph $G$
from Theorem~\ref{thm-strbasic} and let $\TT$ be a tangle in $G$ that points towards
$G_i$ as explained in the previous subsection.  How does the graph $G$ look like from the point of
view of the piece $G_i$?

For a set $S\subseteq V(G)$, an \emph{$S$-bridge} is a subgraph $C$ of $G$ such that either $C$ consists
of a single edge joining vertices of $S$, or $C-S$ is a connected component of $G-S$ and $C$ consists of
this connected component, all its neighbors in $S$, and all the edges between these neighbors and $V(C)\setminus S$.
The set $V(C)\cap S$ is the \emph{attachment set} of the bridge.  Let $D$ be the subgraph of $G$ with vertex set
$V(G)\setminus(V(C)\setminus S)$ and with edge set $E(G)-E(C)$.  Then $(C,D)$ is the \emph{separation defined by
the bridge $C$}; note that $V(C\cap D)$ is precisely the attachment set of the bridge.

Let $A_i$ be the set of apices of $G_i$, and let $G'_i$ be the embedded part of $G_i-A_i$ drawn in the surface $\Sigma_i$.
Each $V(G_i-A_i)$-bridge $C$ of $G-A_i$ attaches either to a clique in a vortex of $G-A_i$, or to $G'_i$.
Note that since $\TT$ points towards $G_i$, the separation $(C,D)$ of $G-A_i$ defined by the bridge belongs to the tangle $\TT-A_i$.

\subsubsection{Embedded part}
Let us first consider the bridges that attach to the embedded part.  Each of them attaches to a clique, which by the assumption
of high representativity of $G'_i$ must have order at most $4$.  Actually, it is even possible to restrict ourselves to attachments
to cliques of order at most $3$ (if a bridge $C$ attaches to a clique $K$ of size $4$,
then let $T$ be the triangle in $K$ that bounds the largest open disk and let $u$ be the vertex of $K$ not belonging to $T$;
we can modify the decomposition so that the embedded part is $G'_i-u$, by adding a clique-sum with $K$ on $T$),
and similarly the attachments to triangles can be restricted only to facial triangles.
Thus, each $V(G_i-A_i)$-bridge of $G-A_i$ that attaches to $G'_i$ can be represented
by a disk in $\Sigma_i$, chosen so that the disks representing different bridges intersect only in the vertices of $G'_i$
shared by the bridges.  This motivates the following definitions.  

A \emph{society} consists of a graph $s$ and a cyclically ordered subset of its vertices denoted by $\partial s$.
If $|\partial s|\le 3$, we say that the society is a \emph{cell}.
A \emph{segregation} of a graph $G$ is a set $S$ of edge-disjoint societies such that
$G=\bigcup_{s\in S} s$ and for all distinct $s,s'\in S$ we have $V(s\cap s')\subseteq \partial s\cap \partial s'$.
An \emph{arrangement} $\alpha$ of a segregation $S$ in a surface $\Sigma$ is a function that to vertices
of $\bigcup_{s\in S} \partial s$ assigns pairwise distinct points in $\Sigma$ and to societies of $S$ assigns
closed disks in $\Sigma$ with pairwise disjoint interiors, such that
\begin{itemize}
\item for all $s\in S$, if $v_1,\ldots, v_k$ is the cyclic ordering of $\partial s$, then the points $\alpha(v_1),\ldots,\alpha(v_k)$
appear in order in the boundary of the disk $\alpha(s)$, and
\item all distinct $s,s'\in S$ satisfy $\alpha(s)\cap \alpha(s')=\{\alpha(v):v\in \partial(s)\cap \partial(s')\}$.
\end{itemize}
Thus, in the setting of the previous paragraph, if $G''_i$ is the union of the $V(G_i-A_i)$-bridges of $G-A_i$ that attach to $G'_i$,
then $G''_i$ has a segregation into cells (corresponding to the bridges) with an arrangement in $\Sigma$.

\subsubsection{Vortices}
Next, let us consider the $V(G_i-A_i)$-bridges of $G-A_i$ that attach to the vortices of the almost-embedding of $G_i-A_i$.

Let $F$ be a vortex of $G_i-A_i$ with boundary sequence $v_1$, \ldots, $v_t$ and path decomposition of width $p-1$ with bags $B_1$, \ldots, $B_t$
in order.  Let $C_1$, \ldots, $C_m$ be the $V(G_i-A_i)$-bridges of $G-A_i$ that attach to $F$.
Recall that in a tree (or path) decomposition of a graph, each clique is contained in one of the bags.  Hence,
the attachment set of each such bridge $C_j$ is contained in one of the bags $B_1$, \ldots, $B_t$; let us choose
such a bag (arbitrarily if there are more bags with this property) and let its index be denoted by $i_j$.
Let $F'=\bigcup_{j=1}^m C_j$, and for $i=1,\ldots, t$, let $B'_i=\bigcup_{j:i_j=i} C_j$.

For a tree or path decomposition, its \emph{adhesion} is the maximum size of the intersection of its distinct bags.
Note that $B'_1$, \ldots, $B'_t$ are bags of a path decomposition of $F'$ in order.  These bags do not necessarily have
bounded size, but since the bridges $C_1$, \ldots, $C_m$ are disjoint except for the parts contained in $F$,
the decomposition has adhesion at most $p$.  This implies that for $1\le i\le t$, $F'$ contains at most $p$
pairwise vertex-disjoint paths with one end in $\{v_1,\ldots, v_i\}$ and the other end in $\{v_{i+1},\ldots,v_t\}$.

It turns out that the last property is sufficient to ensure that $F'$ has such a path decomposition of bounded adhesion
containing $v_1$, \ldots, $v_t$ in its bags in order, and it is easier to work with this property rather
than keeping track of the path decomposition.
Formally, we say that a society $s$ is a \emph{$p$-vortex} if for every partition of the cyclic ordering of $\partial s$
to contiguous intervals $I$ and $J$, there exist at most $p$ pairwise vertex-disjoint paths in $s$ from $I$ to $J$
(or equivalently by Menger's theorem, there exists a separation $(A,B)$ of $s$ of order at most $p$ with $I\subseteq V(A)$ and $J\subseteq V(B)$);
and we have the following correspondence with path decompositions of bounded adhesion.

\begin{lemma}[{Robertson and Seymour~\cite[(8.1)]{rs9}}]\label{lemma-pvortex}
Let $p\ge 0$ be an integer.
Let $s$ be a $p$-vortex and let $v_1$, \ldots, $v_m$ be the cyclic ordering of $\partial s$.
Then $s$ has a path decomposition of adhesion at most $p$ with bags $X_1$, \ldots, $X_m$ in order,
such that $v_i\in X_i$ for $1\le i\le m$.
\end{lemma}

For each $p$-vortex $s$, we will fix arbitrarily a path decomposition satisfying the conditions of Lemma~\ref{lemma-pvortex},
and call it the \emph{standard path decomposition of the $p$-vortex $s$}, and we let its bags be denoted by $X^s_1$, \ldots, $X^s_m$ in
order.  For a vertex $v_i\in \partial s$, we let $X^s(v_i)=\{v_i\}\cup(X^s_i\cap X^s_{i-1})\cup (X^s_i\cap X^s_{i+1})$,
where $X^s_0=X^s_{m+1}=\emptyset$.  Note that when $s$ is obtained from a vortex in an almost-embedding, $X^s(v_i)$ is the set of vertices in
that the bag $X^s_i$ intersects the rest of the graph.

\subsubsection{Segregations and tangles}

As explained in the previous two subsections, we aim to rephrase the structure theorem from the point
of view of a piece $G_i$ of the decomposition that is pointed to by a tangle $\TT$, 
in terms of (an arrangement of) a segregation of $G-A_i$ consisting of cells and $p$-vortices.
How does the fact that the piece is pointed to by $\TT$ reflect in this reformulation?

Each cell $s$ of the segregation corresponds directly to a $V(G_i-A_i)$-bridge of $G-A_i$ that attaches to the embedded
part of the piece $G_i$, and as we already mentioned at the beginning of Subsection~\ref{subs-loc},
this means that the separation $(s,D)$ of $G-A_i$ defined by $s$ (with $V(s\cap D)=\partial s$) belongs to the tangle $\TT-A_i$.
For a $p$-vortex $s$, the situation is a bit more complicated, as the separation $(s,D)$ of $G-A_i$ defined by $s$ may have
arbitrarily large order.  Thus, we need to be a bit more careful, and only forbid $s$ from containing the large part of any
separation of $\TT-A_i$.

Let $\TT'$ be a tangle in a graph $F$ and let $S$ be a segregation of $F$.  We say that $S$ is \emph{$\TT'$-central} if for all $s\in S$,
no separation $(C,D)\in\TT'$ satisfies $D\subseteq s$.  It turns out that for $p$-vortices, the condition of $\TT$-centrality needs to be verified only
for separations of order at most $2p+1$.
\begin{lemma}[{Robertson and Seymour~\cite[(2.1)]{rs15}}]\label{lemma-central}
Let $p\ge 1$ be an integer, let $\TT'$ be a tangle of order at least $5p+2$ in a graph $F$
and let $S$ be a segregation of $F$ such that all societies of $S$ are $p$-vortices.
If no separation $(C,D)\in \TT'$ of order at most $2p+1$ and society $s\in S$ satisfy $D\subseteq s$,
then $S$ is $\TT'$-central.
\end{lemma}
Let us remark that each cell is a $1$-vortex, and thus Lemma~\ref{lemma-central} applies even if
some societies of $S$ are cells.

\subsubsection{Local form of Theorem~\ref{thm-strbasic}}

Putting together all the ingredients we described, we obtain the following form of Theorem~\ref{thm-strbasic}
relative to a tangle.  A segregation $S$ has \emph{type $(p,k)$} if there exists a set $S_0\subseteq S$ of size at most $k$
such that all societies of $S_0$ are $p$-vortices and all societies of $S\setminus S_0$ are cells.

\begin{theorem}[{Roberson and Seymour~\cite[(3.1)]{robertson2003graph}}]\label{thm-strbasic-local}
For any graph $H$, there exist integers $k, p, a, \theta\ge 0$ with the following property.
Let $\TT$ be a tangle of order at least $\theta$ in a graph $G$.
If $H$ is not a minor of $G$, then there exists $A\subseteq V(G)$ with $|A|\le a$ and a $(\TT-A)$-central segregation $S$ of $G-A$
of type $(p,k)$ with an arrangement in a surface $\Sigma$ in which $H$ cannot be embedded.
\end{theorem}

It is relatively easy to derive Theorem~\ref{thm-strbasic} from Theorem~\ref{thm-strbasic-local}; we will explain the procedure
in the following subsection, where we use it to derive Theorem~\ref{thm-strapex} from its local form.

\subsection{Local form of Theorem~\ref{thm-strapex}}

We will now state a variant of Theorem~\ref{thm-strapex} relative to a tangle, extending Theorem~\ref{thm-strbasic-local}.
Given a segregation $S$ consisting only of cells with an arrangement $\alpha$ in a surface $\Sigma$, let $T(S)$ be the multigraph
whose vertex set is  $\bigcup_{s\in S}\partial s$, with edge set consisting of cliques on $\partial s$ for all $s\in S$;
if two vertices belong to several cells, they are joined by the corresponding number of edges.  The graph $T(S)$ is embedded in
$\Sigma$ in the natural way, with the placement of vertices given by $\alpha$ and the edges of the clique on $\partial s$ being
drawn inside the disk $\alpha(s)$ for each $s\in S$.

\begin{theorem}\label{thm-strapex-local}
For any graph $H$, there exist integers $k, p, a, \theta\ge 0$ with the following property.
Let $\TT$ be a tangle of order at least $\theta$ in a graph $G$.
If $H$ is not a minor of $G$, then there exists $A\subseteq V(G)$ with $|A|\le a$ and a $(\TT-A)$-central segregation $S$ of $G-A$
of type $(p,k)$ with an arrangement in a surface $\Sigma$ in which $H$ cannot be embedded, and there exists a set $S_0\subseteq S$
of size at most $k$ such that all societies of $S\setminus S_0$ are cells and
\begin{itemize}
\item[\textrm{(i)}] all but at most $a(H,\Sigma)-1$ vertices of $A$ only have neighbors in $A\cup \bigcup S_0$, and
\item[\textrm{(ii)}] $\bigcup_{s\in S_0} \partial s\subseteq V(T(S\setminus S_0))$ and every triangle in
$T(S\setminus S_0)$ bounds a disk in $\Sigma$ whose interior contains no vertices (only possibly edges parallel
to the edges of the triangle).
\end{itemize}
\end{theorem}

We aim to prove Theorem~\ref{thm-strapex-local}; but first, let us show that it implies Theorem~\ref{thm-strapex}.
Our strategy is to decompose an $H$-minor-free graph $G$ into pieces recursively: We obtain a structure
as in Theorem~\ref{thm-strapex-local} with respect to some tangle and use $T(S\setminus S_0)$ and the standard decompositions
of the $p$-vortices in $S_0$ to form the root piece of the decomposition of $G$.  Then, we process the cells of $S\setminus S_0$
and parts of the $p$-vortices in $S_0$ recursively, obtaining their decompositions which attach to cliques in the root piece.
There is a minor technicality in that e.g. for $s\in S\setminus S_0$, we need to find a decomposition
not of just $s$, but the graph obtained from $s$ by adding a clique on $\partial s$, so that the clique-sum operation assumptions
are satisfied.

The first step of the decomposition is accomplished using the following lemma (where the $\theta$-unbreakable set $Y$ gives the tangle, by Lemma~\ref{lemma-unbreaktangle}).
We say that graphs $G_1, \ldots, G_r\subset G$ and a graph $G_0$ form a \emph{star split of $G$ of adhesion $\kappa$} if
\begin{itemize}
\item $G\subseteq G_0\cup G_1\cup\ldots\cup G_r$,
\item $V(G_i)\cap V(G_j)\subseteq V(G_0)$ for $1\le i\le j\le r$, and
\item $V(G_i)\cap V(G_0)$ induces a clique in $G_0$ of size at most $\kappa$, for $1\le i\le r$.
\end{itemize}

\begin{lemma}\label{lemma-localform}
For any graph $H$, there exist integers $m, d, a, \theta\ge 0$ with the following property.
Let $G$ be a graph and let $Y$ be a $\theta$-unbreakable set in $G$ of size $3\theta-2$.
If $H$ is not a minor of $G$, then there exist graphs $G_1, \ldots, G_r\subseteq G$ and a graph $G_0$ with $Y\subseteq V(G_0)$
forming a star split of $G$ of adhesion $2\theta-1$, such that $Y$ induces a clique in $G_0$
and there exists a surface $\Sigma_0$ and a set $A_0\subseteq V(G_0)$ such that $G_0$, $A_0$ and $\Sigma_0$ satisfy the conditions
\textrm{(a)}--\textrm{(e)} of Theorem~\ref{thm-strapex}.
\end{lemma}
\begin{proof}
Let $k_1$, $p_1$, $a_1$, $\theta_1$ be the constants of Theorem~\ref{thm-strapex-local} applied to $H$.
Let $d=2p_1+1$, $\theta=\max(\theta_1,a_1+d+1)$, $m=k_1+3\theta-2$ and $a=a_1+3\theta-2$.

Let $\TT$ be the set of separations $(C,D)$ of $G$ of order less than $\theta$ such that $|V(C)\cap Y|\le\theta-1$.
By Lemma~\ref{lemma-unbreaktangle}, $\TT$ is a tangle of order $\theta$ in $G$.
By Theorem~\ref{thm-strapex-local} applied with the tangle $\TT$,
there exists a set $A_1\subseteq V(G)$ of size at most $a_1$ and a $(\TT-A_1)$-central segregation $S$ of $G-A_1$ of type $(p_1,k_1)$
with an arrangement in a surface $\Sigma_0$ in which $H$ cannot be embedded, and a set $S_0\subseteq S$ of size at most $k_1$
containing all non-cell elements of $S$ and satisfying the conditions (i) and (ii) from the statement of Theorem~\ref{thm-strapex-local}.
Let $\{G_1, \ldots, G_r\}$ be the set consisting of the following induced subgraphs of $G$:
\begin{itemize}
\item for each cell $s\in S\setminus S_0$, the subgraph of $G$ induced by $V(s)\cup A_s$, where $A_s$ is the set of vertices
of $A_1$ that have a neighbor in $s-\partial s$.
\item for each $p_1$-vortex $s\in S_0$ and its standard path decomposition (see Lemma~\ref{lemma-pvortex}) with bags
$X^s_1$, \ldots, $X^s_{|\partial s|}$ in order, and for $1\le i\le |\partial s|$, the subgraph of $G$ induced by $X^s_i\cup A_1$.
\end{itemize}
Let $T$ be a simple graph obtained from $T(S\setminus S_0)$ by suppressing parallel edges (i.e., removing all but one edge between
each two vertices).
Let $G''_0$ be the graph equal to the union of $T$ and the cliques with vertex sets $X^s(v)$ for each
$p_1$-vortex $s\in S_0$ and each $v\in \partial s$ (let us recall that $X^s(v)$ was defined after Lemma~\ref{lemma-pvortex}
in such a way that $|X^s(v)|\le 2p_1+1=d$).
Observe that $G''_0$ is almost embedded in $\Sigma_0$ with at most $k_1$ vortices of depth at most $d$, and that $V(G''_0)\cap V(G_i)$
induces a clique in $G''_0$ for $1\le i\le r$.
By the condition (ii) of Theorem~\ref{thm-strapex-local}, every triangle in the embedded part $T$ of $G''_0$ bounds a $2$-cell face.

Let $G'_0$ be the graph obtained from $G''_0$ by
\begin{itemize}
\item adding a clique with vertex set $A_1$,
\item for each cell $s\in S\setminus S_0$ adding all edges between the vertices of $A_s$ and the vertices of $\partial s$, and
\item for each $p_1$-vortex $s\in S_0$ and each vertex $v\in \partial s$, adding all edges between $X^s(v)$ and $A_1$.
\end{itemize}
Since $S$ is a segregation of $G-A_1$, it follows that $G\subseteq G'_0\cup G_1\cup\ldots\cup G_r$.
Furthermore, note that $G''_0=G'_0-A_1$ and that for $1\le i<j\le r$, the set $V(G'_0)\cap V(G_i)$ induces a clique in $G'_0$,
and $V(G_i)\cap V(G_j)\subseteq V(G'_0)$.

Let $G_0$ be the graph obtained from $G'_0$ by adding vertices of $Y\setminus V(G'_0)$ and the edges of the clique with vertex set $Y$,
and for $1\le i\le r$ and each vertex $y\in Y\cap V(G_i)$, adding all edges between $y$ and $V(G_i)\cap V(G'_0)$.
This ensures that $V(G_i)\cap V(G_0)$ induces a clique in $G_0$.  We now provide a bound on the size of this clique.  The vertex set of the clique
is $(V(G_i)\cap (V(G'_0)\cup Y))$, and thus its order is at most $|V(G_i)\cap V(G''_0)|+|A_1|+|V(G_i)\cap Y|\le d+a_1+|V(G_i)\cap Y|$.
Hence, we need to bound the last term.
Recall that $G_i-A_1$ is either equal to $\partial s$ for some $s\in S\setminus S_0$, or to $X^s_v$ for some $s\in S_0$ and $v\in\partial s$.
Let $(G_i-A_1,D_i)$ be the separation of $G-A_1$ such that $V((G_i-A_1)\cap D_i)=\partial s$ in the former case, and $V((G_i-A_1)\cap D_i)=X^s(v)$ in
the latter case.  We have $|V((G_i-A_1)\cap D_i)|\le\max(d,3)=d$, and since the order of the tangle $\TT-A_1$ is at least $\theta-a_1>d$,
we have either $(G_i-A_1,D_i)\in \TT-A_1$ or $(D_i,G_i-A_1)\in \TT-A_1$.  The latter would contradict the fact that the segregation $S$ is $(\TT-A_1)$-central,
and thus the former holds.  Since $(G_i-A_1,D_i)\in \TT-A_1$, there exists a separation $(C'_i,D'_i)$ of $G$ such that
$(C'_i,D'_i)\in \TT$, $A_1\subseteq V(C'_i\cap D'_i)$, $C'_i-A_1=G_i-A_1$, and $D'_i-A_1=D_i$.  By the choice of $\TT$,
we have $|V(C'_i)\cap Y|<\theta$, and thus $|V(G_i)\cap Y|<\theta$.  We conclude that the clique $V(G_i)\cap V(G_0)$
has order at most
$$|V(G_i)\cap V(G''_0)|+|A_1|+|V(G_i)\cap Y|<d+a_1+\theta<2\theta.$$
Consequently, $G_1, \ldots, G_r$ together with $G_0$ form a star split of $G$ of adhesion $2\theta-1$.

By the construction, $Y$ induces a clique in $G_0$.  Let $A_0=A_1\cup Y$ and let us argue that $G_0$, $A_0$, and $\Sigma_0$ satisfy the
conditions (a)--(e).  The condition (a) follows by the choice of $\Sigma_0$.  We have $|A_0|\le |A_1|+|Y|\le a_1+3\theta-2=a$, giving (b).

For the condition (c), we start with the almost embedding of $G''_0$ in $\Sigma_0$ with at most $k_1$ vortices of depth at most $d$.
For a vertex $y\in Y\setminus A_1$, let $Z'_y$ be the set of non-boundary vertices of the embedded part $T$ of $G''_0$ that are adjacent to $y$
and do not belong to $Y$. Note that $Z'_y$ can be non-empty only if $y$ is contained either in the embedded part of $G''_0$,
or in $V(s)$ for some $s\in S\setminus S_0$, and in the latter case, $Z'_y\subseteq \partial s$.  In either case, the vertices
of $Z'_y$ are incident with the same face of $T$.
For each $y\in Y\setminus A_1$, select $Z_y\subseteq Z'_y$ such that $\bigcup_{y\in Y\setminus A_1} Z_y=\bigcup_{y\in Y\setminus A_1} Z'_y$
and the sets $Z_y$ for $y\in Y\setminus A_1$ are pairwise disjoint (i.e., remove repeated v from all but one of the sets $Z'_y$).  For each $y\in Y\setminus A_1$ such that $Z_y$ is non-empty,
we can select a closed disk $\Delta_y\subset \Sigma_0$ intersecting $T-Y$ exactly in the vertices of $Z_y$,
so that the disks for distinct vertices of $Y\setminus A_1$ are disjoint, and furthermore they are disjoint from the disks
representing the vortices of $G''_0$.  For each $y\in Y\setminus A_1$ such that $Z_y$ is non-empty, we create a new vortex
of depth $1$ consisting only of the vertices of $Z_y$ (and no edges).  In this way, we obtain an almost-embedding of
$G''_0-Y=G_0-A_0$ in $\Sigma_0$ with at most $k_1+|Y|\le k_1+3\theta-2=m$ vortices of depth at most $d$, as required in (c).

By the condition (ii) of Theorem~\ref{thm-strapex-local}, each triangle in $T$ bounds a $2$-cell face,
and thus each triangle in the embedded part of $G_0-A_0$ (which is equal to $T-Y$) bounds a $2$-cell face as well.
Hence, the condition (d) is satisfied.

Finally, consider the condition (e).  By the condition (i) of Theorem~\ref{thm-strapex-local},
there exists a set $A'_1\subseteq A_1$ of size at most $a(H,\Sigma_0)-1$ such that
vertices of $A_1-A'_1$ are in $G$ only adjacent to vertices of $A_1\cup \bigcup S_0$.
In the construction of $G'_0$, the only new adjacencies between $A_1$ and the embedded part $T$
are formed by edges between $A_s$ and $\partial s$ for cells $s\in S\setminus S_0$, and the definition of $A_s$
implies that $A_s\subseteq A'_1$.  Finally, in the construction of almost-embedding of $G_0-A_0$, we
introduced new vortices to contain the neighbors of the vertices of $Y\setminus A_1$.  Consequently,
among the apices of $A_0$, only those belonging to $A'_1$ can be major, showing that the condition (e) is satisfied.
\end{proof}

Proving Theorem~\ref{thm-strapex} is now just a matter of applying Lemma~\ref{lemma-localform} recursively,
with a minor technicality of dealing with breakable sets.

\begin{proof}[Proof of Theorem~\ref{thm-strapex}, assuming Theorem~\ref{thm-strapex-local}]
Let $m$, $d$, $a_0$ and $\theta$ be the constants of Lemma~\ref{lemma-localform} applied to $H$, and let $a=\max(a_0,4\theta-3)$.
We will prove a stronger claim: for any $H$-minor-free graph $G$ and a set $Y\subseteq V(G)$ of size at most $3\theta-2$,
the graph $G+K_Y$ obtained from $G$ by adding edges of the clique on $Y$ is a clique-sum of graphs $G_1$, \ldots, $G_s$ 
satisfying (with appropriate subsets of their vertices and surfaces) the conditions (a)--(e).
We will prove the claim by induction on the number of vertices of $G$, and thus we assume that the claim
holds for all graphs with less than $|V(G)|$ vertices.

If $|V(G)|<3\theta-2$, then we set $s=1$, $G_1=G+K_Y$, $A_1=V(G)$, and we let $G_1$ be the null surface.
Therefore, suppose that $|V(G)|\ge 3\theta-2$.  Note that if $Y'\supseteq Y$ and $G+K_{Y'}$ has a decomposition
as described in the first paragraph, then $G+K_Y$ has such a decomposition as well (obtained by removing
from the pieces of the decomposition all edges of $E(G+K_Y')\setminus E(G+K_Y)$
not used in clique-sums).  Hence, we can without loss of generality add vertices to $Y$, so that $|Y|=3\theta-2$.

If $Y$ is $\theta$-unbreakable, then let $G'_1,\ldots, G'_r$ and $G'_0$ be the star split of $G$ of adhesion $2\theta-1$
obtained by Lemma~\ref{lemma-localform}.  For $1\le i\le r$, let $Y_i=V(G'_i\cap G'_0)$.  We have $|Y_i|\le 2\theta-1$,
and since $Y\cap V(G'_i)\subseteq Y_i$, it follows that not all vertices of $Y$ belong to $G'_i$.  Consequently, $|V(G'_i)|<|V(G)|$.
By the induction hypothesis, $G'_i+K_{Y_i}$ can be expressed as a clique-sum of pieces satisfying the conditions (a)--(e).
Since $Y_i$ induces a clique in $G'_0$ by the definition of a star split, and since $G'_0$ satisfies the conditions (a)--(e)
and $Y$ induces a clique in $G'_0$ by Lemma~\ref{lemma-localform}, it follows that $G+K_Y$ can be expressed as a clique-sum of
pieces satisfying the conditions (a)--(e).

Finally, consider the case that $Y$ is not $\theta$-unbreakable, and thus there exists a separation $(C_1,C_2)$ of $G$ of order less than $\theta$
such that for $i=1,2$, we have $|V(C_1\cap C_2)|+|Y\setminus V(C_i)|<3\theta-2$.
In particular, $|V(C_i)\cap Y|<|Y|$, and thus $|V(C_i)|<|V(G)|$.  Letting $Y_i=(V(C_1\cap C_2))\cup (Y\setminus V(C_{3-i}))$,
we conclude by the induction hypothesis that $C_i+K_{Y_i}$ can be expressed as a clique-sum of pieces satisfying the conditions (a)--(e).
Let $G_1=K_{Y\cup V(C_1\cap C_2)}$, note that $|V(G_1)|\le |Y|+|V(C_1\cap C_2)|\le 4\theta-3\le a$, and thus
$G_1$ satisfies the conditions (a)--(e) with $A_1=V(G_1)$ and $\Sigma_1$ being the null surface.  Clearly, $Y_1$ and $Y_2$
induce cliques in $G_1$, and thus $G+K_Y$ can be expressed as a clique-sum of
pieces satisfying the conditions (a)--(e).
\end{proof}

\subsection{Distance and minors in embedded graphs}\label{ss-emb}

We would now like to derive Theorem~\ref{thm-strapex-local} from Theorem~\ref{thm-strbasic-local}, the main idea being
that if there were many apex vertices each attaching to many distant parts of the arrangement, we would obtain $H$ as a minor of $G$
(and otherwise we can create new $p$-vortices to cover the points where the apex vertices attach).  However, in addition to the need
to state these nebulous ideas formally, there are two significant issues to deal with.  The first (and easier to fix) is that
the apex vertices may attach to the interiors of the cells in $S\setminus S_0$, and the cells do not necessarily have to be connected,
making such attachments worthless in obtaining a minor; we will deal with this issue in the next subsection simply by requiring
that the cells are connected, and splitting disconnected cells into several subcells.
Here, we will consider the second obstacle: the arrangement of the segregation in the surface does not need to imply existence of any
substantial minors.

To see what the problem is, let us first consider a simpler case of minors of graphs embedded in surfaces.  A \emph{simple closed curve}
in a surface $\Sigma$ is a subset of $\Sigma$ homeomorphic to the circle (while a \emph{simple curve} is a subset homeomorphic to a
line segment).  A simple closed curve $c$ in $\Sigma$ is \emph{contractible}
if there exists a closed disk in $\Sigma$ whose boundary is equal to $c$, and \emph{non-contractible} otherwise.  Given a graph $G$ embedded
in $\Sigma$, a curve or a closed curve $c$ is \emph{$G$-normal} if $c$ intersects $G$ only in the vertices of $G$.
If $\Sigma$ is not the sphere, the \emph{representativity}
of $G$ is the minimum possible number of intersections between $G$ and a $G$-normal non-contractible simple closed curve.
Note that every minor of $G$ has at most as large representativity as $G$.  Hence, in order to allow a graph $H$ drawn in $\Sigma$
to be a minor of $G$, we need to ensure that the representativity of $G$ is at least as large as the representativity of $H$.

Importantly for us, Robertson and Seymour~\cite{rs7} proved a rough converse to this observation.

\begin{theorem}\label{thm-repr}
Let $H$ be a graph and let $\Sigma$ be a surface other than the sphere.
If $H$ can be drawn in $\Sigma$, then there exists an integer $r_0$ such that every graph $G$
drawn in $\Sigma$ with representativity at least $r_0$ contains $H$ as a minor.
\end{theorem}

Thus, it would be nice to enhance Theorem~\ref{thm-strbasic-local} to guarantee that the arrangement of $S$ (or more precisely,
the embedding of the graph $T(S\setminus S_0)$) has large representativity.  This is indeed possible, and we will do
so in the following subsection, but let us first discuss two caveats: the case that $\Sigma$ is the sphere, and the fact that we
actually need a claim stronger than Theorem~\ref{thm-repr}.

The case that $\Sigma$ is the sphere is excluded from the statement of Theorem~\ref{thm-repr}, since the definition of representativity
is meaningless in that case.  Instead, in that case we need to assume that $G$ contains a large grid, which by Lemma~\ref{lemma-dem}
is equivalent with $G$ having large tree-width, and that in turn is equivalent with the existence of a tangle of large order in $G$
(Robertson and Seymour~\cite[(5.2)]{rs10}).  It is convenient to be able to handle both the sphere and the non-sphere case uniformly,
and fortunately large representativity is also equivalent with the existence of a certain tangle of large order.

Suppose that $G$ is embedded in a surface $\Sigma$ different from the sphere with representativity $\theta$, and consider any $G$-normal simple closed curve $c$
that intersects $G$ in less than $\theta$ points.  Then $c$ is contractible, and thus it bounds a closed disk $\Delta\subset\Sigma$.
Consequently, $c$ defines a separation $(A_c,B_c)$ of $G$ of order less than $\theta$ such that $A_c$ is the subgraph of $G$ drawn in $\Delta$ and $B_c$ is the
subgraph of $G$ drawn in the closure of $\Sigma\setminus\Delta$.  Thus, the embedding assigns a ``large part'' ($B_c$---the one still using
substantial part of $\Sigma$, while $A_c$ is planar) to each such separation.  Of course, not all separations of $G$ are of this form,
but it turns out that this still uniquely determines a tangle.
\begin{theorem}[{Robertson and Seymour~\cite[(4.1)]{rs11}}]\label{thm-reprtangle}
Let $G$ be a graph embedded in a surface $\Sigma$ different from the sphere with representativity $\theta$.
There exists a unique tangle $\TT$ of order $\theta$ in $G$ such that $(A_c,B_c)\in \TT$ for all $G$-normal simple closed curves $c$
that intersect $G$ in less than $\theta$ points.
\end{theorem}

This motivates the following definition: Let $G$ be a graph embedded in a surface $\Sigma$ (possibly the sphere)
and let $\TT$ be a tangle in $G$ of order $\theta$.  We say that $\TT$ is \emph{respectful} if for every $G$-normal simple closed curve $c$
in $\Sigma$ that intersects $G$ in less than $\theta$ points, there exists a closed disk $\Delta\subseteq \Sigma$ bounded by $c$
such that $(G\cap \Delta,G\cap \overline{\Sigma\setminus\Delta})\in \TT$;
let this disk be denoted by $\ins(c)$ (or by $\ins_\TT(c)$, when the tangle is not uniquely determined by the context).
Observe that if $\Sigma$ is not the sphere, then the fact that $\TT$ is respectful implies
that every non-contractible $G$-normal simple closed curve intersects $G$ in at least $\theta$ points, and thus the representativity of $G$
is at least $\theta$.  On the other hand, if $\Sigma$ is the sphere, then the condition that $\TT$ is respectful is trivially true.
Hence, Theorem~\ref{thm-repr} can be generalized to include the sphere case as follows.
\begin{theorem}\label{thm-resp}
Let $H$ be a graph and let $\Sigma$ be a surface.
If $H$ can be drawn in $\Sigma$, then there exists an integer $\theta_0$ such that every graph $G$
drawn in $\Sigma$ that has a respectful tangle of order at least $\theta_0$ contains $H$ as a minor.
\end{theorem}

The second caveat we mentioned is that we actually need more than just this result.  In our setting, we have apex vertices
that attach to various points in the embedded graph, and we want to use them to obtain the minor we seek.  This requires
us to find a \emph{rooted minor} in the embedded graph, with roots corresponding to the attachments of the apex vertices.
Let $H$ and $G$ be graphs and let $f:V(H)\to V(G)$ be a partial injective function.  We say that $H$ is an \emph{$f$-rooted minor} of $G$
if we can assign to vertices $v\in V(H)$ pairwise vertex-disjoint connected subgraphs $\mu(v)$ of $G$, and assign to edges $e\in E(H)$
distinct edges $\mu(e)\in E(G)$, such that
\begin{itemize}
\item if $e=uv\in E(H)$, then $\mu(e)$ has one end in $\mu(u)$ and the other end in $\mu(v)$, and
\item if $v\in \dom(f)$, then $f(v)\in V(\mu(v))$.
\end{itemize}
The assignment $\mu$ is a \emph{model} of the $f$-rooted minor $H$ in $G$ (for unrooted minors, a model
is defined in the same way, excluding the last condition).

We would like to strengthen Theorem~\ref{thm-resp} to state that if $f$ maps vertices in $\dom(f)$ to sufficiently ``distant'' vertices
of $G$, then $H$ is an $f$-rooted minor of $G$.  What does ``distant'' mean in this context?  Robertson and Seymour~\cite{rs11} showed that
respectful tangles naturally give rise to a metric in the embedded graph.  Before we define this metric, it is convenient to introduce
the concept of a radial graph.

An embedding of a graph is \emph{2-cell} if all its faces are $2$-cell.
Let $G$ be a graph with a $2$-cell embedding in a surface $\Sigma$, and let $\TT$ be a respectful tangle in $G$ of order $\theta\ge 3$.
A \emph{radial graph} $\R(G)$ is a bipartite graph drawn in $\Sigma$ obtained as follows: The vertex set of $\R(G)$
consists of $V(G)$, and for each face $f$ of $G$, a vertex $v_f$ drawn inside $f$.  If $v_1v_2\ldots v_k$ is the facial
walk of $f$, then $\R(G)$ contains edges $v_fv_1$, \ldots, $v_fv_k$ drawn inside $f$ in the natural way.  For a vertex $v\in V(G)$,
let $\rad(v)=v$ be the corresponding vertex of $\R(G)$.  For a face $f$ of $G$, let $\rad(f)=v_f$.  Finally, note that
each edge $e\in E(G)$ is contained in a unique face of $\R(G)$ of length $4$, and let $\rad(e)$ denote this face.
By an \emph{atom} of $G$, we mean a vertex, an edge, or a face of $G$, and let $\atom(G)$ denote the set of atoms of $G$.

Each cycle $K$ in $\R(G)$ traces a simple closed curve $c$ that intersects $G$ in $|K|/2$ vertices; thus, if $|K|<2\theta$,
we can define $\ins(K)=\ins(c)$.  Next, we extend $\ins$ to other subgraphs of $\R(G)$: An \emph{$\ell$-restraint} is a connected subgraph $W$ of $\R(G)$
containing a closed walk of length less than $2\ell$ that traverses all edges of $W$.  If $\ell<\theta$ and $W$ is an $\ell$-restraint (so all cycles
in $W$ have length less than $2\theta$), we define $\ins(W)$ as the union of $W$ and of $\ins(K)$ for all cycles $K$ of $W$.

Finally, we are ready to define the metric: for any atoms $x$ and $y$ of $G$, let $d_\TT(x,y)=0$ if $x=y$, let $d_\TT(x,y)=\ell$
if $x\neq y$ and $\ell<\theta$ is the smallest integer such that $\R(G)$ contains an $\ell$-restraint $W$ with $\rad(x)\cup\rad(y)\subseteq\ins(W)$,
and let $d_\TT(x,y)=\theta$ if no such $\ell$-restraint exists.  Let us remark that it need not be immediately obvious to the reader that $d$ is indeed a metric;
this was proved in~\cite[(9.1)]{rs11}.  Note that the metric has some connection to the usual graph distance: if vertices $x,y\in V(G)$ are connected
by a path $P$ of length $n$ in $G$, then the union of the boundaries of the faces $\rad(e)$ for $e\in E(P)$ forms a $2n$-restraint showing
that $d_\TT(x,y)\le 2n$.  However, there are other reasons why two vertices of $G$ could be near, e.g., if they are both contained in $\ins(K)$
for a short cycle $K\subseteq \R(G)$.  Nevertheless, in a metric derived from a respectful tangle of large order, some atoms
are far apart.

\begin{lemma}[{Roberson and Seymour~\cite[(8.12)]{rs11}}]\label{lemma-far}
Let $G$ be a graph with a $2$-cell embedding in a surface and let $\TT$ be a respectful tangle in $G$ of order $\theta\ge 1$.
For every atom $a$ of $G$, there exists an edge $e$ of $G$ such that $d_\TT(a,e)=\theta$.
\end{lemma}

With the issue of the metric out of the way, we can state the result concerning the existence of rooted minors in embedded graphs,
which appears implicitly in~\cite{rs12} (see (3.2) and its application in the proofs of (4.3), (4.4), and (4.5)---Theorem~\ref{thm-gmembed}
can be proved in the same way).

\begin{theorem}\label{thm-gmembed}
For every surface $\Sigma$ and a graph $H$ that can be drawn in $\Sigma$, there exists an integer $\theta_0$ such that the following holds.
Let $G$ be graph $2$-cell embedded in $\Sigma$ with a respectful tangle $\TT$ of order at least $\theta_0$.
Let $f:V(H)\to V(G)$ be a partial injective function. If all distinct vertices $x,y\in \dom(f)$ satisfy
$d_\TT(f(x),f(y))\ge\theta_0$, then $H$ is an $f$-rooted minor of $G$.
\end{theorem}

\subsection{Cleaning up the structure}

We would now like to enhance Theorem~\ref{thm-strbasic-local} so that the arrangement in the surface (or more precisely, the graph $T(S\setminus S_0)$)
is guaranteed to have a respectful tangle of large order, so that Theorem~\ref{thm-gmembed} can be applied.  It is natural to ask for the
tangle $\TT-A$ from the statement of Theorem~\ref{thm-strbasic-local} to have this property.  Indeed, this is possible, with a minor caveat:
as we already mentioned, we will be working in the graph $T(S\setminus S_0)$ different from $G-A$, and thus we must must actually consider
a different (but related) tangle in $T(S\setminus S_0)$.  Another point is that using Theorem~\ref{thm-gmembed}, we will obtain a minor in $T(S\setminus S_0)$,
which we would like to be also a minor of $G-A$.  Hence, it is natural to ask that $T(S\setminus S_0)$ is a minor of $G-A$, which requires adding further
assumptions on the cells of $S\setminus S_0$.

This motivates the following definitions.  Let $F$ be a minor of a graph $G$, and let $\mu$ be a model of $F$ in $G$.
Let $\TT_F$ be a tangle of order $\theta\ge 2$ in $F$.  Let us define $\TT_G$ as the set of separations $(A,B)$ of $G$ of
order less than $\theta$ such that there exists a separation $(A',B')$ of $F$ belonging to $\TT_F$ satisfying
$\mu(E(F))\cap E(A)=\mu(E(A'))$ (i.e., the edges of the minor of $F$ in $G$ that are contained in $A$ are
precisely the edges of $A'$).  Then $\TT_G$ is a tangle in $G$ of order $\theta$ (see~\cite[(6.1)]{rs10}).
If $\TT$ is a tangle in $G$ of order at least $\theta$ such that $\TT_G\subseteq \TT$, we say that the tangle $\TT$ is
\emph{conformal} with $\TT_F$.  Note that the conformality is defined with respect to a specific model $\mu$ of $F$ in $G$
(and indeed, different models may give rise to different tangles in $G$); however, the model will usually be clear from the context
and we will omit mentioning it.

Let $s$ be a cell.  We say that $s$ is \emph{linked} if $s$ is connected, $\partial s\neq\emptyset$, and
letting $K_s$ denote the clique with vertex set $\partial s$ and $\id$ the identity function from $V(K_s)$ to $V(S)$,
there exists an $\id$-rooted minor of $K_s$ in $s$.  We say that a segregation $S$ consisting only of cells is \emph{linked}
if all its cells are linked.  Note that in that case, $T(S)$ is a minor of $\bigcup S$.

If $\rho$ is an arrangement of a segregation $S$ in some surface and $S_0$ is a subset of $S$ such that all elements of $S\setminus S_0$
are cells, then note that for any $s\in S_0$, the graph $T(S\setminus S_0)$ has a unique face containing the interior of the disk $\rho(s)$.
We will call this face \emph{the face of $T(S\setminus S_0)$ containing $s$}.

We can now state the enhanced form of Theorem~\ref{thm-strbasic-local}.

\begin{theorem}\label{thm-strbasic-local-better}
For any graph $H$, there exist integers $k,p\ge 0$ such that for any non-decreasing positive function $\phi$ of one variable,
there exist integers $\theta>\alpha\ge 0$ with the following property.
Let $\TT$ be a tangle of order at least $\theta$ in a graph $G$.  If $H$ is not a minor of $G$, then
there exists $A\subseteq V(G)$ with $|A|\le \alpha$ and a $(\TT-A)$-central segregation $S$ of $G-A$
of type $(p,k)$, which has an arrangement in a surface $\Sigma$ in which $H$ cannot be embedded, and there exists $S_0\subseteq S$ 
of size at most $k$ such that all societies of $S\setminus S_0$ are cells and
\begin{itemize}
\item[\textrm{(i)}] the segregation $S\setminus S_0$ is linked,
\item[\textrm{(ii)}] the embedding of $T(S\setminus S_0)$ is $2$-cell and contains a respectful tangle $\TT'$ of order at least $\phi(|A|)$ conformal with $\TT-A$, and
\item[\textrm{(iii)}] if $f_1$ and $f_2$ are faces of $T(S\setminus S_0)$ containing distinct $p$-vortices of $S_0$, then $d_{\TT'}(f_1,f_2)\ge\phi(|A|)$.
\end{itemize}
\end{theorem}

The enhancements of Theorem~\ref{thm-strbasic-local-better} are basically just restatements of some of the results of~\cite{rs17}, and we will provide more
details shortly; however, it does not seem possible to state the derivation without many further technical definitions with no further use in this paper,
and thus this will be an ``expert-only'' explanation not accessible to anyone not already familiar with the Robertson-Seymour series of papers.
Here, let us provide an intuition of how Theorem~\ref{thm-strbasic-local-better} can be derived from Theorem~\ref{thm-strbasic-local}.

For the condition (i), we can without loss of generality assume that the cells are connected, since otherwise we can split disconnected cells
into several subcells and arrange them in the surface in a natural way inside the disk corresponding to the original cell.
Without loss of generality, the set $S_0$ is non-empty (since every cell is also a $1$-vortex).  Hence, if some cell $s\in S\setminus S_0$ has
empty boundary, then we can make it a part of one of the $p$-vortices of $S_0$.  So, if a cell $s\in S\setminus S_0$ is not linked,
it has $|\partial s|=3$ and it does not contain triangle on $\partial s$ as a rooted minor; but then it is easy to see that $s$ contains a $1$-cut
separating a vertex of $\partial s$ from the other two, and we can split $s$ into two subcells on this $1$-cut.  Repeating these reductions,
we ensure that (i) holds.

For the condition (ii), it is easy to derive a conformal tangle $\TT'$ in $T(S\setminus S_0)$ of large order from the tangle $\TT-A$ in $G-A$,
using the fact that the segregation $S$ is $(\TT-A)$-central and that the elements of $S_0$ are $p$-vortices.  If $\Sigma$ is the sphere,
then $\TT'$ is automatically respectful.  Suppose that $\Sigma$ is not the sphere.  If the embedding of $T(S\setminus S_0)$ in $\Sigma$ has
large representativity, then by Theorem~\ref{thm-reprtangle} it contains a unique respectful tangle, and using the fact that the segregation $S$ is $(\TT-A)$-central,
we can see that $\TT'$ is this tangle\footnote{This is actually a bit of an oversimplification, as $\TT'$ could also be a non-respectful tangle
derived from a planar grid embedded in the surface and separated from the part of the embedding with large representativity
by a small cut; but let us ignore this case, which can be dealt with similarly to the described one.}.  Hence, the only problem would be if there existed a $T(S\setminus S_0)$-normal non-contractible simple closed curve $c$
intersecting $T(S\setminus S_0)$ in only a few vertices; let $A_1$ denote the set of these vertices.  But then $G-(A\cup A_1)$ has a segregation of type $(p,k)$
with an arrangement in the simpler surface obtained from $\Sigma$ by deleting $c$ and capping the resulting hole(s) by disk(s) (with minor technical complications if $\Sigma-c$ is disconnected),
and we repeat the reasoning.  As with each repetition the genus of the surface decreases (and the genus of $\Sigma$ is bounded since $H$ cannot be drawn in it),
the number of repetitions is bounded, and in each repetition we only add a bounded number $|A_1|$ of new apex vertices, and thus in the end the number of apex vertices will still be bounded,
as required.  The condition (iii) is argued about similarly.

Let us remark at this point about an additional feature of Theorem~\ref{thm-strbasic-local-better} we did not mention so far: the possibility to prescribe
the dependence $\phi(|A|)$ of the order of the tangle $\TT'$ on the number $|A|$ of apex vertices.  Why is this needed?  We could of course in the
process described in the previous paragraph aim just for some fixed bound $\phi$ on the representativity of the embedding of $T(S\setminus S_0)$ and
the corresponding respectful tangle $\TT'$.  But then in each repetition of the argument we may be introducing up to $\phi-1$ new apex vertices,
and in the end we could have many more apex vertices than $\phi$, which would complicate our further arguments (it is much easier to deal with the case
that we have a very small number of apex vertices attaching to a huge graph in the surface).  Fortunately, we can easily amend the procedure from the previous paragraph
to increase our requirement for the representativity of the embedding depending on how many apex vertices we already got.

\subsubsection{($\star$) Cleaning up the structure: the details}\label{sssec-clear}

For the formal argument, we use the following result, which is essentially (13.4) of~\cite{rs17}. 

\begin{theorem}\label{thm-gm17}
For any graph $H$, there exist integers $p$ and $q$ such that for any non-decreasing positive function $\sigma$ of one variable,
there exist integers $\theta>z\ge 0$ with the following property.
Let $\TT$ be a tangle of order at least $\theta$ in a graph $G$ controlling no $H$-minor
of $G$. Then there exists $A\subseteq V(G)$ with $|A|\le \alpha$ and a
true $\sigma(|A|)$-redundant $(\TT-A)$-central portrayal of $G-A$ with warp $\le\!p$ and at most $q$ cuffs in a surface in which $H$ cannot be embedded.
\end{theorem}

Let us remark that there are several differences between the statements of Theorem~\ref{thm-gm17} and of (13.4) in~\cite{rs17}.
\begin{itemize}
\item In (13.4), there is a different order of the quantifiers meaning that $p$ and $q$ depend also on $\sigma$ and not only on $H$.
However, an inspection of their choices in the first paragraph of the proof of (13.4) shows that they are independent on $\sigma$,
and thus the order of quantifiers in Theorem~\ref{thm-gm17} is correct.
\item Furthermore, the function $\sigma$ in (13.4) may have two additional parame\-ters---$p$ and the surface of the portrayal.
Here, we chose a simpler formulation without this dependence. However, since $p$ only depends on $H$ and there are only
finitely many choices of the surface (also depending only on $H$), we can maximize $\sigma$ over the possible choices of the parameters,
showing that our formulation is not significantly weaker.
\end{itemize}

\begin{proof}[Proof of Theorem~\ref{thm-strbasic-local-better}]
Choose integers $p_0$ and $q$ so that Theorem~\ref{thm-gm17} for $H$ is satisfied with $p,q$ replaced by $p_0,q$.  Let $k=\max(q,1)$ and $p=2p_0+2$. 
Let $\sigma(x)=\max(x,\phi(x))+4p+5$ and let $\theta$ and $\alpha$ be as in Theorem~\ref{thm-gm17}.

We apply Theorem~\ref{thm-gm17} to $G$, obtaining a set $A\subseteq V(G)$ and a true $\sigma(|A|)$-redundant $(\TT-A)$-central portrayal $\pi$ of $G-A$
in a surface $\Sigma'$ in which $H$ cannot be embedded.
Let $\Sigma$ be the surface without boundary obtained from $\Sigma'$ by, for each cuff $\Theta$ of $\Sigma'$, adding an open disk with boundary $\Theta$ disjoint
with $\Sigma'$.  We turn the portrayal $\pi$ into
an arrangement of a segregation $S$ of $G-A$ of type $(p,k)$ in $\Sigma$ by replacing each cuff $\Theta$ by a $p$-vortex consisting of
the union of all the graphs in the border cells of $\Theta$.
Let $S_0$ be the set of non-cell $p$-vortices of $S$ if it is non-empty, and the set consisting of one cell of $S$ otherwise.
Let us note that $S$ is $(\TT-A)$-central by (2.1) of~\cite{rs15} and (4.3) of~\cite{rs17}, since $\TT-A$ has order at least $4p+2$.
Note that by (9.1) and (9.2) of~\cite{rs17}, $S\setminus S_0$ can be made linked by, for each cell $s$, moving all components of $s$ that
whose intersection with $\partial s$ is empty to one of the societies of $S_0$.

The embedding of $T(S\setminus S_0)$ in $\Sigma$ is $2$-cell by (8.1) of~\cite{rs17} (the faces of $T(S\setminus S_0)$ containing the societies of $S_0$ are $2$-cell by (8.3)).
Let us define a tangle $\TT'$ in $T(S\setminus S_0)$ of order $\phi(|A|)$ as follows.
By (6.1) and (6.5) of~\cite{rs11}, it suffices to define an even slope $\ins$ of order $\phi(|A|)$ in the radial drawing of $T(S\setminus S_0)$.
Let $c$ be a simple closed $T(S\setminus S_0)$-normal curve intersecting $T(S\setminus S_0)$ in less than $\phi(|A|)$ vertices which corresponds to a cycle in the radial drawing of $T(S\setminus S_0)$.
Note that since $\pi$ is $\sigma(|A|)$-redundant, (6.3) of~\cite{rs17} implies that $c$ intersects at most one face containing a society of $S_0$.
If $c$ intersects such face, then by (6.4) of~\cite{rs17}, there exists a disk $\Delta\subseteq \Sigma'$ bounded by $c$ such that the part of the portrayal $\pi$
inside $c$ is small in the tangle $\TT-A$.  In this case, we set $\ins(c)=\Delta$.   Suppose now that $c$ intersects a face $f$ containing a $p$-vortex $s\in S_0$, and let $v_1$ and $v_2$ be the vertices
in the intersection of $c$ and the boundary of $f$.  Let $\Theta$ be the corresponding cuff of $\Sigma'$.
For $i\in \{1,2\}$, if $v_i\in V(s)$, then let $\Delta_i=\emptyset$.
If $v_i\not\in V(s)$, then let $\Delta_i$ be a closed disk in $\Sigma'\cap \overline{f}$ intersecting the boundary of $f$ in two vertices---$v_i$ and another
vertex $v'_i\in V(s)$---such that $c$ and $\Delta_i$ intersect in a simple curve contained in the boundary of $\Delta_i$.  Furthermore, choose $\Delta_1$ and $\Delta_2$ so that they are disjoint.
Let $c'$ be the closed curve given as the symmetric difference of $c$ and the boundary of $\Delta_1\cup \Delta_2$.  By (6.3) of~\cite{rs17} applied to the I-arc
$c'\cap \Sigma'$, there exists a disk $\Delta'$ whose boundary is contained in $\Theta\cup (c'\cap \Sigma')$ such that the part of the portrayal $\pi$
inside $c$ is small in the tangle $\TT-A$.  Let $\Delta''$ be the closure of the symmetric difference of $\Delta'$ and $\Delta_1\cup \Delta_2$.
Let $\Delta$ be the closed disk contained in $\Delta''\cup (\Sigma\setminus\Sigma')$ bounded by $c$.  We set $\ins(c)=\Delta$.
Note that the choice of $\Delta_i$ (and even $v'_i$) is not necessarily unique, but it is easy to see that all possible choices give the same value of $\ins(c)$.
Observe that $\ins$ is an even slope, and the corresponding tangle $\TT'$ in $T(S\setminus S_0)$ is conformal with $\TT-A$.
Furthermore, (6.3) and (6.4) of~\cite{rs17} imply that if $f_1$ and $f_2$ are faces of $T(S\setminus S_0)$ containing distinct $p$-vortices of $S_0$,
then $d_{\TT'}(f_1,f_2)\ge\phi(|A|)$.
\end{proof}

\subsection{Localizing the apices}

Let $G$ be a graph, $A$ a subset of its vertices, and let $S$ be a linked segregation of a subgraph $G'\subseteq G-A$ into
cells with an arrangement $\rho$ in surface $\Sigma$.  Consider any cell $s\in S$ and let $\repr(s)$ be the atom of $T(S)$ chosen as follows:
\begin{itemize}
\item if $|\partial s|=3$, then let $\repr(s)$ be the face of $T(S)$ bounded by the triangle on $\partial s$ drawn in $\rho(s)$,
\item if $|\partial s|=2$, then let $\repr(s)$ be the edge of $T(S)$ joining the vertices of $\partial s$ drawn in $\rho(s)$, and
\item if $|\partial s|=1$, then let $\repr(s)$ be the vertex of $\partial s$.
\end{itemize}
We say that $\repr(s)$ is the \emph{atom representing $s$}.
If $S$ is a subset of larger segregation $S_1$ with an arrangement in $\Sigma$ extending $\rho$, then for
$s\in S_1\setminus S$ we define $\repr(s)$ to be the face of $T(S)$ containing $s$.

Let $\TT'$ be a respectful tangle in $T(S)$ of order at least $\phi$.
For an integer $n$,
we say that a vertex $v\in A$ is a \emph{$(\phi,n)$-major apex} if there exist cells $s_1,\ldots, s_n\in S$
such that $v$ has a neighbor in each of $s_1$, \ldots, $s_n$ and $d_{\TT'}(\repr(s_i),\repr(s_j))\ge \phi$ for $1\le i<j\le n$.

Let us now prove the previously advertised claim that if at least $a(H,\Sigma)$ of the apices of $A$ are
sufficiently major, then a graph $H$ is a minor of $G$.

\begin{lemma}\label{lemma-majorminor}
For every graph $H$ and a surface $\Sigma$, there exists an integer $\phi$ as follows.
Let $G$ be a graph, $A$ a subset of its vertices, and let $S$ be a linked segregation of a subgraph $G'\subseteq G-A$ into
cells with an arrangement in $\Sigma$.  Let $\TT'$ be a respectful tangle in $T(S)$ of order at least $2\phi$.
If at least $a(H,\Sigma)$ vertices of $A$ are $(2\phi,2|E(H)|)$-major apices, then $H$ is a minor of $G$.
\end{lemma}
\begin{proof}
Let $B$ be a set of $a(H,\Sigma)$ vertices of $H$ such that $B-H$ can be drawn in $\Sigma$.
Let $H'$ be the graph obtained from $H$ subdividing each edge twice.  We interpret $B$ in natural
way as a subset of vertices of $H'$, and observe that $B$ is an independent set in $H'$ and the vertices
of $B$ have no common neighbors.
Let $\theta_0$ be the constant of Theorem~\ref{thm-gmembed} applied for $\Sigma$ and $H'-B$,
and let $\phi=\theta_0+4$.

Let $A_1\subseteq A$ be a set of $(2\phi,2|E(H)|)$-major apices of size $a(H,\Sigma)$, and let $g_0:B\to A_1$ be an arbitrary bijection.
Let $N$ be the set of neighbors of vertices of $B$ in $H'$, and note that $|N|\le 2|E(H)|$.  
We can greedily find an injective function $g_1:N\to S$ with the following properties:
\begin{itemize}
\item For every vertex $v\in N$ adjacent to a vertex $w\in B$ in $H'$, the apex $g_0(w)$ has a
neighbor in $g_1(v)$, and
\item distinct $u,v\in N$ satisfy $d_{\TT'}(\repr(g_1(u)),\repr(g_1(v)))\ge \phi$.
\end{itemize}
Indeed, suppose we already chose $g_1\restriction N'$ for some $N'\subsetneq N$,
and consider a vertex $v\in N\setminus N'$ adjacent to a vertex $w\in B$ in $H'$.  Since $g_1(w)$ is
a $(2\phi,2|E(H)|)$-major apex, it has neighbors in cells $s_1,\ldots, s_{|N|}\in S$
such that $d_{\TT'}(\repr(s_i),\repr(s_j))\ge 2\phi$ for $1\le i<j\le |N|$.  For every atom $x\in \atom(T(S))$
there exists at most one $i\in\{1,\ldots,|N|\}$ such that $d_{\TT'}(x,\repr(s_i))<\phi$.
Since $|N|>|N'|$, there exists $i\in\{1,\ldots,|N|\}$ such that $d_{\TT'}(\repr(g_1(z)),\repr(s_i))\ge\phi$
for all $z\in N'$, and we can set $g_1(v)\colonequals s_i$.

Consider now for some vertex $v\in N$ adjacent to $w\in B$ in $H'$ the cell $s=g_1(v)$.
Let $w'$ be a neighbor of $g_0(w)$ in $s$.
Since $s$ is linked, there exists a rooted model $\mu$ of a clique on
$\partial s$ as a minor of $s$.  Furthermore, $s$ is connected, and thus we can assume that
$\{V(\mu(r)):r\in\partial s\}$ is a partition of $V(s)$; let $r_v\in\partial s$ denote the vertex
such that $w'\in\mu(r_v)$.  Let $T_1$ be the graph obtained from $T(S)$ by adding vertices
of $A$ and edges $r_vg_0(w)$ for all adjacent $w\in B$ and $v\in N$, and observe that $T_1$
is a minor of $G$.

Let $g:N\to V(T(S))$ be defined by setting $g(v)=r_v$ for all $v\in N$.
Since $g(v)$ is incident with $\repr(g_1(v))$, we have $d_{\TT'}(\repr(g_1(v)),g(v))\le 2$,
and thus $d_{\TT'}(g(u),g(v))\ge \phi-4\ge \theta_0$ for distinct $u,v\in N$.
By Theorem~\ref{thm-gmembed}, $H'-B$ is a $g$-rooted minor of $T(S)$.
Consequently, $H'$ is a minor of $T_1$, and thus also a minor of $G$.
Since $H$ is a minor of $H'$, the claim of the lemma follows.
\end{proof}

Lemma~\ref{lemma-majorminor} bounds the number of major apices, and the other apices
only attach to a bounded number of areas of bounded radius.  We can even select the
areas far apart from one another by possibly combining the nearby ones to a single
larger area.  To show this, we use the following claim (which is standard, but we include its proof for completeness).

\begin{lemma}\label{lemma-mkdist}
For all integers $\theta_0,n>0$ and every non-decreasing positive function $f$, there exists an integer $\theta$ such that the following holds.  Let $Z$ and $U$ be sets of points of a metric space with metric $d$,
such that $|Z|\le n$ and for every $u\in U$ there exists $z\in Z$ with $d(u,z)<\theta_0$.  Then, there exists a subset $Z'\subseteq Z$ and an integer $t\le \theta$ such that
\begin{itemize}
\item for every $u\in U$, there exists $z\in Z'$ with $d(u,z)<t$, and
\item for distinct $z_1,z_2\in Z'$, $d(z_1,z_2)\ge f(t)$.
\end{itemize}
Furthermore, if the elements of a set $Z''\subseteq Z$ are at distance at least $\theta$ from each other, then we can choose $Z'$ so that $Z''\subseteq Z'$.
\end{lemma}
\begin{proof}
For $1\le i\le n-1$, let $\theta_i=\theta_{i-1}+f(\theta_{i-1})$; and set $\theta=\theta_{n-1}$.  We construct a sequence of sets $Z=Z_0\supset Z_1\supset \ldots \supset Z_{n'}$ with $n'<n$
such that for every $u\in U$ and $i\le n'$, there exists $z\in Z_i$ with $d(u,z)<\theta_i$, as follows: suppose that we already found $Z_i$.  If $d(z_1,z_2)\ge f(\theta_i)$
for all distinct $z_1,z_2\in Z_i$, then set $n'=i$ and stop.  Otherwise, there exist distinct $z_1,z_2\in Z_i$ such that $d(z_1,z_2)<f(\theta_i)$ and $z_2\not\in Z''$;
in this case, set $Z_{i+1}=Z_i\setminus \{z_2\}$.  Clearly, the set $Z'=Z_{n'}$ has the required properties for $t=\theta_{n'}$.
\end{proof}

Let us now apply this lemma to neighborhoods of non-major apices.

\begin{lemma}\label{lemma-localapi}
For every graph $H$, surface $\Sigma$, integers $a,k\ge 0$ and a non-decreasing positive function $f$,
there exists an integer $\theta\ge 0$ as follows.
Let $G$ be a graph, $A$ a subset of its vertices of size at most $a$, and let $S$ be a segregation of $G-A$
with an arrangement in $\Sigma$.  Let $S_0$ be a subset of $S$ of size at most $k$ such that all societies
of $S\setminus S_0$ are linked cells.  Let $\TT'$ be a respectful tangle in $T(S\setminus S_0)$ of order at least $\theta$
such that all faces $f_1$ and $f_2$ containing distinct societies of $S_0$ satisfy $d_{\TT'}(f_1,f_2)\ge \theta$.
If $H$ is not a minor of $G$, then there exists an integer $t\le \theta$ and
sets $A_0\subseteq A$ of size at most $a(H,\Sigma)-1$ and $S_1\subseteq S$ of size at most
$k+2a|E(H)|$ such that $S_0\subseteq S_1$ and
\begin{itemize}
\item the vertices of $A\setminus A_0$ only have neighbors in societies $s\in S$
such that there exists $s_1\in S_1$ with $d_{\TT'}(\repr(s),\repr(s_1))<t$, and
\item $d_{\TT'}(\repr(s_1),\repr(s_2))\ge f(t)$ for all distinct $s_1,s_2\in S_1$.
\end{itemize}
\end{lemma}
\begin{proof}
Let $\phi$ be the constant of Lemma~\ref{lemma-majorminor} for $H$ and $\Sigma$.
Let $\theta$ be the constant of Lemma~\ref{lemma-mkdist} applied with $\theta_0=2\phi$, $n=k+2a|E(H)|$
and the function $f$.

Let $A_0\subseteq A$ consist of the apices that are $(2\phi,2|E(H)|)$-major (with $G'=\bigcup (S\setminus S_0)$
and its segregation $S\setminus S_0$).  By Lemma~\ref{lemma-majorminor}, since $H$ is not a minor of $G$ we have $|A_0|\le a(H,\Sigma)-1$.
For each $v\in A\setminus A_0$, let $S_v\subseteq S\setminus S_0$ be a maximal set such that
$v$ has a neighbor in each cell of $S_v$ and $d_{\TT'}(\repr(s_1),\repr(s_2))\ge 2\phi$ for all distinct $s_1,s_2\in S_v$.
Since $v$ is not $(2\phi,2|E(H)|)$-major, we have $|S_v|<2|E(H)|$.  Furthermore, by the maximality of $S_v$,
if $s\in S\setminus S_0$ is a cell containing a neighbor of $v$, then $d_{\TT'}(\repr(s),\repr(s_1))<2\phi$ for some $s_1\in S_v$.

Let $Z=S_0\cup \bigcup_{v\in A\setminus A_0} S_v$ and let $U$ be the set of societies $s\in S$ containing a neighbor of some
vertex of $A\setminus A_0$.  Then we can apply Lemma~\ref{lemma-mkdist} with $Z''=S_0$ and let $S_1$ be the resulting set $Z'$.
\end{proof}

\subsection{Extending the $p$-vortices}

Our plan now is to turn each of the areas of bounded radius found in Lemma~\ref{lemma-localapi}
into a $p'$-vortex for some $p'$, and in this subsection we discuss the technicalities of the process.

Firstly, we claim that each such area of small radius is contained
in an open disk in the surface whose radius is not much larger and whose boundary forms a cycle in the graph.
Let $T$ be a graph with a $2$-cell embedding in a surface $\Sigma$ and let $\TT'$ be a respectful tangle in $T$.
If $a$ is an atom of $T$ and $t$ is a positive integer, then a \emph{$t$-zone around $a$} is an open disk $\Lambda\subset \Sigma$
bounded by a cycle $C\subseteq T$ such that $a\subseteq\Lambda$ and $d_{\TT'}(a,a')\le t$ for all atoms $a'$ of $T$
contained in $\overline{\Lambda}$.

\begin{lemma}[{Robertson and Seymour~\cite[(9.2)]{rs14}}]\label{lemma-zone}
Let $T$ be a graph with a $2$-cell embedding in $\Sigma$ and let $\TT'$ be a respectful tangle in $T$ of order $\theta$.
Let $a$ be an atom of $T$.  For every integer $t$ with $2\le t\le \theta-3$, there exists a $(t+2)$-zone $\Lambda$ around $a$
such that every $a'\in\atom(T)$ with $d_{\TT'}(a,a')<t$ satisfies $a'\subseteq \Lambda$.
\end{lemma}

For a zone $\Lambda$, the subgraph of $T$ obtained by removing all vertices and edges of $T$ drawn in $\Lambda$ is said to be
obtained by \emph{clearing $\Lambda$}.  Note that since the zone $\Lambda$ is bounded by a cycle which is clearly contractible,
if the embedding of $T$ is $2$-cell, then the embedding of the graph obtained by clearing $\Lambda$ is also $2$-cell.
Quite naturally, clearing a zone of bounded radius does not
alter the distances much (the metric in the subgraph obtained by clearing is defined as based
on a conformal tangle in the subgraph).

\begin{lemma}[{Robertson and Seymour~\cite[(7.10)]{rs12}}]\label{lemma-clearing}
Let $T$ be a graph with a $2$-cell embedding in a surface and let $\TT'$ be a respectful tangle in $T$ of order $\theta\ge 4t+3$.
Let $\Lambda$ be a $t$-zone around some atom of $T$ and let $T'$ be the graph obtained from $T$ by clearing $\Lambda$.
Then, there exists a unique respectful tangle $\TT''$ in $T'$ of order $\theta-4t-2$ such that
whenever $a',b'$ are atoms of $T'$ and $a,b$ are atoms of $T$ with $a\subseteq a'$ and $b\subseteq b'$, we have
$d_{\TT'}(a,b)-4t-2\le d_{\TT''}(a',b')\le d_{\TT'}(a,b)$.  Furthermore, $\TT''$ is conformal with $\TT'$.
\end{lemma}

Let $\TT$ be a tangle in a graph $G$ and let $S$ be a $\TT$-central segregation of $G$,
let $S_0\subseteq S$ be a set of size at most $k$ such that all elements of $S\setminus S_0$ are
linked cells and all elements of $S_0$ are $p$-vortices,
and suppose that $S$ has an arrangement $\rho$ in a surface $\Sigma$ such that the embedding of $T=T(S\setminus S_0)$ is $2$-cell.
Let $\TT'$ be a respectful tangle of order $\theta$ in $T$ conformal with $\TT$.

Consider an arbitrary society $s\in S$ and let $f=\repr(s)$ be the atom of $T$ representing it.
Let $t\ge 2$ be an integer such that $\theta>12t+2$, and $d_{\TT'}(f_1,f_2)>12t+2$ for any faces $f_1$ and $f_2$ of $T$
containing different societies of $S_0$, and $d_{\TT'}(f,f_1)>12t+2$ for any face $f_1$ containing a society of $S_0$ distinct from $s$.
Recall that regardless of whether $s$ is a cell or whether $s\in S_0$, the society $s$ is a $p$-vortex.
In this setting, we are now going to show how to create a $(3t+4p)$-vortex containing $s$ and all cells of $S\setminus S_0$ such that the atoms
of $T$ that represent them are at distance at most $t$ from $f$.

First, if there exists a simple closed $T$-normal curve $c$ intersecting $T$ in at most $t$ vertices such that $f\subseteq \ins_{\TT'}(c)$, then
choose such a closed curve $c$ with $\ins_{\TT'}(c)\cap T$ maximal and let $\Delta_0=\ins_{\TT'}(c)$; otherwise, let $\Delta_0$ be the closure of $f$.
Note that $d_{\TT'}(f, a)\le t$ for every atom $a\in \atom(T)$ contained in $\Delta_0$, since in the case that $\Delta_0=\ins_{\TT'}(c)$, the curve
$c$ can be shifted slightly inside the faces of $T$ to trace a closed walk of length at most $2t$ in $\R(T)$.

Next, let $R$ be the set of vertices $v\in T$ drawn in $\overline{\Sigma\setminus\Delta_0}$ such that there exists a simple $T$-normal curve $c$ intersecting $T$ in less than $t$ points
starting in $v$ and ending in a vertex $w$ of $T$ contained in the boundary of $\Delta_0$.
Observe that the distance in the metric $d_{\TT'}$ from $v$ to $w$ is less than $2(t-1)$, and
thus $d_{\TT'}(f,v)<3t-2$.  By Lemma~\ref{lemma-zone}, there exists a $3t$-zone $\Lambda'\subset \Sigma$ around $f$ in $T$ such that $R\subset \Lambda'$.
Note that all vertices of $R$ are contained in a single face $g$ of $T-R$, and since $\Lambda'$ is bounded by a cycle $C'$ in $G$, the face $g$ is a subset of $\Lambda'$.

\begin{figure}
\center{\includegraphics{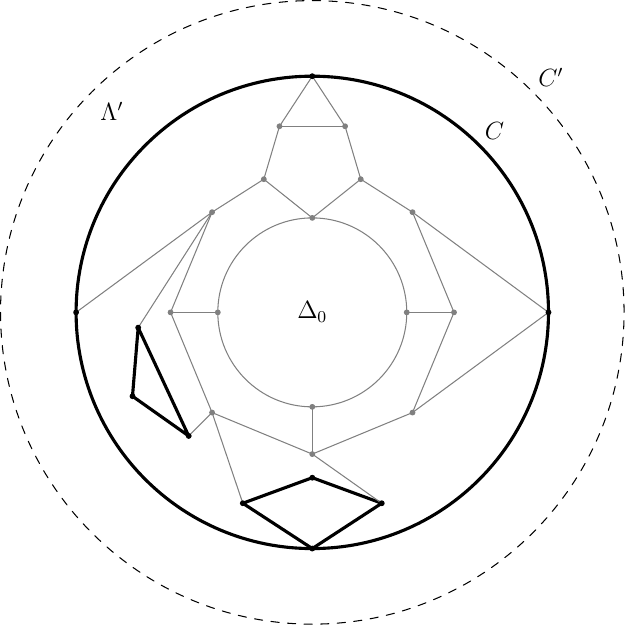}}
\caption{Face obtained by deleting vertices of $R$ (drawn in grey), with $t=3$.}\label{fig-rface}
\end{figure}

The face $g$ is not necessarily $2$-cell, but it either contains $f$ as a subset or separates $f$ from the boundary cycle $C'$ of $\Lambda'$.
Let $W$ be the boundary walk of $g$ contained in the component of $T-R$ that contains $C'$; note that $W$ is not necessarily a cycle, but it separates
$f$ from $C'$, and thus there exists a cycle $C\subseteq W$ that also separates $f$ from $C'$, see Figure~\ref{fig-rface}.
The open disk $\Lambda\subseteq\Lambda'$ bounded by $C$ contains $g$.
Let $S_1\subseteq S\setminus S_0$ consist of the cells $s_1$ such that the atom $\repr(s_1)$ is a subset of $\Lambda$.
Let $s'$ be the society with graph $s\cup\bigcup S_1$ and boundary $\partial s'$ consisting of the vertices of $C$ in order.
Let $S'=\{s'\}\cup S\setminus (\{s\}\cup S_1)$, and let $S'_0=\{s'\}\cup S_0\setminus \{s\}$.
Note that $S'$ is a segregation of $G$ with an arrangement $\rho'$ in $\Sigma$ that matches the arrangement $\rho$ on $S'\setminus\{s'\}$,
with the disk $\rho'(s')$ contained in the union of $\Lambda$ with the points representing $V(C)=\partial s'$.   We say that
$S'$, $S'_0$, $s'$, and $\rho'$ are obtained from $S$, $S_0$, $s$, and $\rho$ by a \emph{$t$-extension at $s$}.
Observe that $T(S'\setminus S'_0)$ is a subgraph of $T(S\setminus S_0)$, and since its face containing $s'$ is bounded by the cycle $C$, it follows that its embedding in $\Sigma$ is $2$-cell.

\begin{lemma}\label{lemma-vortext}
Let $p\ge 1$, $t\ge 2$ and $\phi\ge 18t+20p+2$ be integers.
Let $\TT$ be a tangle of order at least $\phi$ in a graph $G$ and let $S$ be a $\TT$-central segregation of $G$,
let $S_0\subseteq S$ be a set such that all elements of $S\setminus S_0$ are linked cells and all elements of $S_0$ are $p$-vortices,
and suppose that $S$ has an arrangement $\rho$ in a surface $\Sigma$ such that the embedding of $T=T(S\setminus S_0)$ is $2$-cell.
Let $\TT'$ be a respectful tangle of order $\phi$ in $T$ conformal with $\TT$.
Consider a society $s\in S$ and let $f=\repr(s)$ be the atom of $T$ representing it.
Suppose that $d_{\TT'}(f_1,f_2)>12t+2$ for any faces $f_1$ and $f_2$ of $T$ containing distinct $p$-vortices of $S_0$,
and $d_{\TT'}(f,f_1)>12t+2$ for any face $f_1$ containing a $p$-vortex of $S_0$ distinct from $s$.
Let $S'$, $S'_0$, $s'$, and $\rho'$ be obtained from $S$, $S_0$, $s$, and $\rho$ by a $t$-extension at $s$,
and let $T'=T(S'\setminus S'_0)$.
The following claims hold.
\begin{enumerate}
\item The society $s'$ is a $(3t+4p)$-vortex.
\item The embedding of $T'$ has a respectful tangle $\TT''$ of order $\phi-12t-2$ conformal with $\TT'$ (and thus also with $\TT$).
\item All $s_1,s_2\in S'\setminus\{s'\}$ satisfy $$d_{\TT'}(\repr(s_1),\repr(s_2)-12t-2\le d_{\TT''}(\repr(s_1),\repr(s_2))\le d_{\TT'}(\repr(s_1),\repr(s_2),$$
and $$d_{\TT'}(\repr(s_1),\repr(s)-12t-2\le d_{\TT''}(\repr(s_1),\repr(s'))\le d_{\TT'}(\repr(s_1),\repr(s).$$
\item The segregation $S'$ of $G$ is $\TT$-central.
\item If $s_1\in S\setminus S_0$ is a cell such that the atom $f_1=\repr(s_1)$ of $T$ representing it satisfies $d_{\TT'}(f,f_1)\le t$, then $s_1\subseteq s'$.
\end{enumerate}
\end{lemma}
\begin{proof}
Let $\Delta_0$, $\Lambda'$, and $\Lambda$ denote the disks, $R$ the set, and $g$ the face from the construction of the $t$-extension.

For the first claim, consider any partition of $\partial s'$ to two contiguous intervals
$I$ and $J$.  Let $v_1$ and $v_2$ be the endpoints of $I$.  By the choice of $s'$ and $\rho(s')$, for $i\in\{1,2\}$
there exists a simple $T$-normal curve drawn in $\rho(s')$ connecting $v_i$ with a vertex $w_i$ contained in the
boundary of $\Delta_0$ and intersecting $T$ in at most $t$ vertices.  Let $Z_0=(c_1\cup c_2)\cap T$.
If $\Delta_0=\rho(s)$, then let $Z=Z_0\cup X^{s}(w_1)\cup X^{s}(w_2)$, otherwise let $Z$ consist of $Z_0$ and the vertices
of $T$ drawn in the boundary of $\Delta_0$.
Observe that $Z$ separates $I$ from $J$ in $s'$ and that $|Z|\le 2t+\max(4p,t)\le 3t+4p$.  Consequently, $s'$ contains
at most $3t+4p$ pairwise vertex-disjoint paths from $I$ to $J$.  Since the choice of $I$ and $J$ was arbitrary,
this implies that $s'$ is a $(3t+4p)$-vortex.

Recall now that $\Lambda'$ is a $3t$-zone around $f$, and thus $\Lambda\subseteq \Lambda'$ is
a $3t$-zone around $f$ as well.  Since $T'$ is obtained from $T$ by clearing the zone $\Lambda$,
the second and third claims follow from Lemma~\ref{lemma-clearing}.

We use Lemma~\ref{lemma-central} to prove the fourth claim.  Since all societies in $S'$ are $(3t+4p)$-vortices
by the first claim, it suffices to consider a separation $(C,D)\in\TT$ of order at most $6t+8p+1$.
For $s_1\in S'\setminus\{s'\}$, we have $D\not\subseteq s_1$ since $s_1$ is contained
in the segregation $S$ which is $\TT$-central.  Let us now consider the society $s'$.
Note that $T'$ is a minor of $G$ with a model $\mu$ that uses no edges of $s'$, and $\TT$ is conformal
with the tangle $\TT''$ of $T'$.  By the second claim, the order of $\TT''$ is $\phi-12t-2\ge 6t+8p+2$
which is larger than the order of $(C,D)$.  By the definition of conformality, there exists a separation
$(C',D')$ of $T'$ belonging to $\TT''$ of such that $\mu(E(T'))\cap E(C)=\mu(E(C'))$ and $\mu(E(T'))\cap E(D)=\mu(E(D'))$.
Since $(C',D')\in\TT''$, Lemma~\ref{lemma-T2e} implies $E(D')\neq \emptyset$.
Since the model $\mu$ uses no edges of $s'$, the edges $\mu(E(D'))\subseteq E(D)$ do not belong to $s'$,
and thus $D\not\subseteq s'$.  As the choice of $(C,D)$ was arbitrary, Lemma~\ref{lemma-central} implies that $S'$
is $\TT$-central.

Let us now consider the final claim.  By the construction of $s'$, it suffices to show that $f_1$ is a subset of $\Lambda$.
Since $d_{\TT'}(f,f_1)\le t<\phi$, there exists a $t$-restraint $W\subseteq \R(T)$ with $\rad(f)\cup\rad(f_1)\subseteq \ins(W)$.
Observe that either $W$ intersects $f$, or $W$ contains a cycle $W'$ of length at most $2t$
(corresponding to a simple closed $T$-normal curve intersecting $T$ in at most $t$ vertices) with $f\subseteq\ins(W')$.
In either case, the choice of $\Delta_0$ implies that $W$ intersects $\Delta_0$.  Each two vertices of $W$ are joined by
a path of length less than $2t$, corresponding to a $T$-normal simple curve intersecting $T$ in less than $t$ vertices;
and thus $V(W)\subseteq R$.

If $W$ does not intersect $f_1$, then $W$ contains a cycle $W''$ of length at most $2t$
with $f_1\subseteq\ins(W'')$; the simple closed curve tracing $W''$ is drawn in the face $g$.  Note that
$\ins(W'')\subseteq\Lambda$, as otherwise $\Lambda\cup \ins(W)=\Sigma$, giving $d_{\TT'}(f,e)\le 3t<\phi$ for all $e\in E(T)$ in contradiction to
Lemma~\ref{lemma-far}.  Consequently, $f_1\subseteq \Lambda$.  
If $W$ intersects $f_1$, then either $f_1$ is a vertex and $f_1\in R\subseteq \Lambda$, or $f_1$ is a face.
In the latter case, $W$ passes through a vertex of $R$ incident with $f_1$, implying that $f_1\subseteq g\subseteq \Lambda$.
Hence, in all the cases, we have $f_1\subseteq\Lambda$ as required.
\end{proof}

\subsection{Putting the pieces together}

We now apply the vortex extension operation to the result of Lemma~\ref{lemma-localapi}, restricting the
attachments of non-major apices to vortices.

\begin{lemma}\label{lemma-restrapex}
For every graph $H$, surface $\Sigma$ in that $H$ cannot be embedded and for all integers $p_0\ge 1$, $k_0,a\ge 0$,
there exist integers $p$, $k$ and $\phi$ as follows.
Let $G$ be a graph with a tangle $\TT$ and let $A$ be a subset of $V(G)$ with $|A|\le a$.  Suppose that $\TT$
has order at least $\phi+|A|$ and let $S'$ be a $(\TT-A)$-central segregation
of $G-A$ with an arrangement $\rho$ in $\Sigma$.  Let $S'_0\subseteq S'$ be a set of size at most $k_0$
such that all societies of $S'_0$ are $p_0$-vortices and all elements of $S'\setminus S'_0$ are linked cells.
Let $\TT'$ be a respectful tangle in $T'=T(S'\setminus S'_0)$ of order at least $\phi$ conformal with $\TT-A$,
such that all faces $f_1$ and $f_2$ containing distinct societies of $S'_0$ satisfy $d_{\TT'}(f_1,f_2)\ge \phi$.
If $H$ is not a minor of $G$, then there exists a $(\TT-A)$-central segregation $S$ of $G-A$ of
type $(p,k)$ with an arrangement in $\Sigma$, and a set $S_0\subseteq S$ of size at most $k$
such that all societies of $S\setminus S_0$ are cells, satisfying the conditions (i) and (ii)
of Theorem~\ref{thm-strapex-local}.
\end{lemma}
\begin{proof}
Let $k=k_0+2a|E(H)|$ and $f(t)=(12t+2)(k+1)+1$.  Let $\theta$ be the constant of Lemma~\ref{lemma-localapi} applied for $H$, $\Sigma$, $a$, $k_0$ and $f$.
Let $p=4p_0+3\theta$ and $\phi=18\theta+20p+2+(12\theta+2)k$.

Let $t\le\theta$, $A_0\subseteq A$, and $S_1\subseteq S'$ be obtained by Lemma~\ref{lemma-localapi} applied to $S'$ and $S'_0$.
Let $S_1=\{s_1,\ldots,s_n\}$, where $n\le k$.  Let $\TT^0=\TT'$.
Let $(S^0,S^0_0,\rho^0)=(S',S'_0,\rho)$, and for $1\le i\le n$, let $S^i$, $S^i_0$, $s^i$, and $\rho^i$ be obtained
from $S^{i-1}$, $S^{i-1}_0$, $s_i$, and $\rho^{i-1}$ by a $t$-extension at $s_i$.  Note that by Lemma~\ref{lemma-vortext},
$s^i$ is a $p$-vortex, $T^i=T(S^i\setminus S^i_0)$ has a respectful tangle $\TT^i$ of order $\phi-(12t+2)i\ge 18t+20p+2$ conformal
with $\TT-A$, and the segregation $S^i$ of $G-A$ is $(\TT-A)$-central.  Letting $s^{i,j}=s_j$ if $j>i$ and $s^{i,j}=s^j$
if $j\le i$, Lemma~\ref{lemma-vortext} also implies that $d_{\TT^i}(\repr(s^{i,j_1}),\repr(s^{i,j_2}))\ge f(t)-(12t+2)i\ge 12t+3$
for all distinct $j_1,j_2\in\{1,\ldots, n\}$.  Finally, all cells $s\in S'\setminus S'_0$ such that
$d_{\TT'}(\repr(s),\repr(s_i))\le t$ (and thus also $d_{\TT^{i-1}}(\repr(s),\repr(s_i))\le t$) satisfy $s\subseteq s^i$.

According to the outcome of Lemma~\ref{lemma-localapi}, we conclude that all neighbors of vertices of $A\setminus A_0$
in $G-A$ belong to $\bigcup_{i=1}^n s^i$.  Hence, $S^n$ with its subset $S^n_0=\{s^1,\ldots,s^n\}$ satisfy the condition (i) of Theorem~\ref{thm-strapex-local}.
Since $\TT^n$ is a respectful tangle of order greater than three, every triangle $C$ in $T^n$ bounds a disk $\ins(C)\subset \Sigma$.
Let $L$ be the set of interiors of all inclusionwise-maximal disks $\ins(C)$ for triangles $C\in T^n$ such that the interior of $\ins(C)$ contains
a vertex of $T^n$.
Let $\{S_\Lambda:\Lambda\in L\}$ be a partition of $S_L=\{s\in S^n:\repr(s)\subseteq \bigcup L\}$ such that $\repr(s)\subseteq \Lambda$
for all $\Lambda\in L$ and $s\in S_\Lambda$.  For $\Lambda\in L$, let $s_\Lambda=\bigcup S_\Lambda$ be a cell with $\partial s_\Lambda$ consisting
of the vertices of the triangle bounding $\Lambda$.  Let $S=(S^n\setminus S_L)\cup \{s_\Lambda:\Lambda\in L\}$
and let $S_0=(S^n_0\setminus S_L)\cup \{s_\Lambda:\Lambda\in L, S_\Lambda\cap S^n_0\neq\emptyset\}$.
This ensures that $S$ and $S_0$ satisfy both conditions (i) and (ii) of Theorem~\ref{thm-strapex-local}.
\end{proof}

We are now ready to combine Theorem~\ref{thm-strbasic-local-better} and Lemma~\ref{lemma-restrapex} to
finish the argument by proving the local structure result.

\begin{proof}[Proof of Theorem~\ref{thm-strapex-local}]
Let $p_0$ and $k_0$ be the constants of Theorem~\ref{thm-strbasic-local-better} for $H$.
Let $p_1$, $k_1$ and $\phi_1$ be positive non-decreasing functions in variable $a_0$ such that $p_1(a_0)$,
$k_1(a_0)$ and $\phi_1(a_0)$ are greater or equal to the corresponding constants given by Lemma~\ref{lemma-restrapex}
applied to $H$, $p_0$, $k_0$, $a_0$ and for each surface $\Sigma$ in that $H$ cannot be embedded.
Let $\alpha$ and $\theta_1$ be the constants given by Theorem~\ref{thm-strbasic-local-better} for $H$ and the function $\phi_1$.
Let $k=k_1(\alpha)$, $p=p_1(\alpha)$, $a=\alpha$, and $\theta=\theta_1+\alpha$.  Theorem~\ref{thm-strapex-local} then follows by applying
Lemma~\ref{lemma-restrapex} to the segregation obtained by Theorem~\ref{thm-strbasic-local-better}.
\end{proof}

\bibliographystyle{siam}
\bibliography{apexalg}

\baselineskip 11pt
\vfill
\noindent
This material is based upon work supported by the National Science Foundation.
Any opinions, findings, and conclusions or
recommendations expressed in this material are those of the authors and do
not necessarily reflect the views of the National Science Foundation.
\end{document}